\documentclass[12pt,journal,a4paper,draftcls,oneside,onecolumn]{IEEEtran}
\usepackage{times} 
\usepackage{graphicx}
\usepackage{cite}
\usepackage{citesort}
\usepackage{color}
\usepackage{psfrag}
\usepackage{subfigure}
\usepackage{amssymb}
\usepackage{lmodern}
\usepackage{epsfig}
\usepackage{pifont}
\usepackage{amsmath}
\DeclareMathOperator*{\argmax}{argmax}
\usepackage{cases}
\usepackage{array}
\usepackage{multicol}
\usepackage{multirow}
\usepackage[T1]{fontenc}
\usepackage{comment}
\usepackage{enumerate}
\usepackage{amsthm}
\usepackage{indentfirst}
\usepackage{setspace}
\usepackage{cases}
\usepackage[ruled, lined, linesnumbered, commentsnumbered, longend]{algorithm2e}
\SetKwInOut{KwIn}{Input}
\SetKwInOut{KwOut}{Output}
\SetKwProg{Init}{[Initialization]}{}{end}
\SetKwProg{QT}{[QT Iteration]}{}{end}
\SetKwProg{DT}{[DT Iteration]}{}{end}
\SetKwProg{LT}{[Bisection Iteration of $\chi$]}{}{end}
\SetKwProg{QB}{[QT and Bisection Iteration]}{}{end}
\graphicspath{{figures/}}
\theoremstyle{remark}
\newtheorem{remark}{Remark}
\makeatletter
\renewcommand{\maketag@@@}[1]{\hbox{\m@th\normalsize\normalfont#1}}
\makeatother

\renewenvironment{thebibliography}[1]{
  \begin{oldthebibliography}{#1}
    \setlength{\itemsep}{-0.1em}
    \setlength{\parskip}{-0.1em}
}
{
  \end{oldthebibliography}
}

\newtheorem{lemma}{Lemma}

\makeatletter
\def\ifundefined{\@ifundefined}
\makeatother 

\begin{document}

\title{MDD-Enabled Two-Tier Terahertz Fronthaul in Indoor Industrial Cell-Free Massive MIMO}

\author{Bohan Li, Diego Dupleich, Guoqing Xia, Huiyu Zhou, Yue Zhang, Pei Xiao, {\em Senior Member, IEEE}, Lie-Liang Yang, {\em Fellow, IEEE} \thanks{B. Li and H. Zhou are with the School of CMS, Uni. of Leicester, LE1 7RH Leicester, UK (E-mail: bl204,hz143@leicester.ac.uk). D. Dupleich is with Fraunhofer Institute for Integrated Circuits IIS, 91058 Erlangen, Germany  (E-mail: diego.dupleich@iis.fraunhofer.de). G. Xia is with the School of Engineering, Uni. of Leicester, LE1 7RH Leicester, UK (E-mail: gx21@leicester.ac.uk). Y. Zhang is  with the Institute of Communication Measurement Technology, 610095 Chengdu, China (e-mail:
yuezhang@icsmcn.cn). P. Xiao is with 5GIC \& 6GIC, Uni. of Surrey, GU2 7XH Guildford, UK (Email: p.xiao@surrey.ac.uk). L. Yang is with the School of ECS, Univ. of Southampton, SO17 1BJ Southampton, UK. (E-mail: lly@ecs.soton.ac.uk)}}

\maketitle

\begin{abstract}
To make indoor industrial cell-free massive multiple-input multiple-output (CF-mMIMO) networks free from wired fronthaul, this paper studies a multicarrier-division duplex (MDD)-enabled two-tier terahertz (THz) fronthaul scheme. More specifically, two layers of fronthaul links rely on the mutually orthogonal subcarreir sets in the same THz band, while access links are implemented over sub-6G band. The proposed scheme leads to a complicated mixed-integer nonconvex optimization problem incorporating access point (AP) clustering, device selection, the assignment of subcarrier sets between two fronthaul links and the resource allocation at both the central processing unit (CPU) and APs. In order to address the formulated problem, we first resort to the low-complexity but efficient heuristic methods thereby relaxing the binary variables. Then, the overall end-to-end rate is obtained by iteratively optimizing the assignment of subcarrier sets and the number of AP clusters. Furthermore, an advanced MDD frame structure consisting of three parallel data streams is tailored for the proposed scheme. Simulation results demonstrate the effectiveness of the proposed dynamic AP clustering approach in dealing with the varying sizes of networks. Moreover, benefiting from the well-designed frame structure, MDD is capable of outperforming TDD in the two-tier fronthaul networks. Additionally, the effect of the THz bandwidth on system performance is analyzed, and it is shown that with sufficient frequency resources, our proposed two-tier fully-wireless fronthaul scheme can achieve a comparable performance to the fiber-optic based systems. Finally, the superiority of the proposed MDD-enabled fronthaul scheme is verified in a practical scenario with realistic ray-tracing simulations.
\end{abstract}


\IEEEpeerreviewmaketitle

\section{Introduction}\label{section:MDDTHz:intro}
Compared with the conventional co-located massive multiple-input multiple-output (mMIMO), cell-free mMIMO (CF-mMIMO), reaping the merits of coordinated multipoint (CoMP) transmission and cloud radio access network (C-RAN), is capable of achieving more consistent and robust quality of services (QoS) with densely distributed access points (APs) \cite{demir2021foundations}. In order to facilitate CF-mMIMO, there are two main operation modes, namely centralized mode, where central processing unit (CPU) carries out all the signal processing tasks such as channel estimation, coordinated beamforming and data detection, and distributed mode, where each AP equipped with simple baseband processors operates independently and processes everything except the final data detection. Unsurprisingly, the former catches more attention due to its higher achievable rates, but it heavily relies on the fronthaul capacity for sending (collecting) pre-processed data (raw data) to (from) APs. Additionally, no matter which operation mode is applied, the synchronization between APs and users has to be conducted over fronthaul links. Therefore, fronthaul links are of paramount importance in CF-mMIMO systems.

The classic architecture of centralized CF-mMIMO was studied in \cite{ngo2017cell}, where all densely distributed APs are connected with the CPU via wired fronthaul links and each AP simultaneously serves all users. Although it was shown to outperform the conventional small-cell systems, the wired fronthaul incurs high cost, rendering it impractical especially when the number of APs and users are significantly large. In order to relieve the signaling burden of fronthaul links, the authors in \cite{bjornson2020scalable} proposed a scalable solution based on user-centric dynamic cooperation clustering for both centralized and distributed CF-mMIMO systems. In addition, various user-centric clustering (UCC) approaches for mitigating the fronthaul overhead have been also proposed, e.g, in \cite{buzzi2017cell} and \cite{ammar2021downlink}. However, the UCC only considers addressing the problem of low fronthaul capacity, while the huge infrastructure cost of wired fronthaul like fiber-optic cables are neglected. In order to avoid the excessive use of cables, authors in \cite{zhang2022user} proposed a hierarchical wire-based topology, where the CPU only connects with a leading AP within each AP cluster, and the leading AP then forwards the fronthaul signal to the remaining APs in the cluster. Although this topology can avoid overlength cables and reduce CPU overhead, it exhibits the poor flexibility when the extra APs or users have to be installed or served. Moreover, the conception of radio stripe systems was recently proposed in \cite{interdonato2019ubiquitous,shaik2020cell}. Different from the conventional wire-based star topology, radio stripe systems can accommodate all the APs using one stripe, and therefore significantly lower the cost of cabling. However, the installation of radio stripes are restricted by the physical structure of network scenarios, leading to the fact that the APs in the CF systems cannot be freely deployed. In addition, the radio stripes or any other wired cables are not easily built up/tear down, especially in a complicated indoor scenario. 

In order to provide high flexibility without sacrificing the end-to-end system performance, the high-band techniques relied wireless fronthaul have recently caught the attention of researchers. For instance, due to the small wavelength and large bandwidth of millimeter-wave (mmWave), the mmWave-based fronthaul is capable of providing the sufficient fronthaul capacity for ultra-dense networks \cite{gao2015mmwave,stephen2017joint}. In \cite{demirhan2022enabling}, the authors applied mmWave-enabled point-to-point fronthaul between the CPU and all the APs in CF-mMIMO. To further reduce the expenditure of infrastructure, similar to \cite{zhang2022user}, a hybrid two-tier fronthaul scheme was also proposed in \cite{demirhan2022enabling}, where the CPU transmits fronthaul signal to the leading AP of each cluster over the mmWave band, and the set of APs within cluster are connected with wired cables. However, the wired connections still inevitably limit the flexibility of AP clustering. That is to say, once the AP clusters are determined, it is not easy to readjust them when the number and also the locations of users have changed. 

\subsection{Motivations}
To date, the works on wireless fronthaul in CF-mMIMO are very limited, possibly due to the following two challenges: (i) As APs are densely distributed, the wireless fronthaul links between the CPU and APs should provide sufficiently high capacity and wide coverage; (ii) The conventional TDD-based fronthaul fails to fully exploit the time-frequency resources, and hence leading to the degradation of end-to-end performance. On the other hand, compared with outdoor CF-mMIMO, the indoor industrial CF-mMIMO systems are even more in need of wireless fronthaul, from the perspective of convenient upgrade and low deployment cost. In fact, with the evolution of smart industry, CF-mMIMO has been deemed as the most promising technique to realize the reliable and fast wireless services for industrial communications \cite{elhoushy2021cell}. In \cite{zhang2022user,alonzo2021cell,peng2022resource,wang2020implementation}, the indoor industrial CF-mMIMO with wired fronthaul have been investigated, but to the best of our knowledge, no previous research has studied the indoor industrial CF-mMIMO employed with wireless fronthaul. 

Against this background, in this paper, we propose a multicarrier-division duplex (MDD)-enabled two-tier terahertz (THz) fronthaul scheme, aiming to address the above-mentioned challenges for indoor industrial CF-mMIMO systems. More specifically, inspired by \cite{zhang2022user} and \cite{demirhan2022enabling}, a two-tier fully-wireless fronthaul scheme, consisting of the links between the CPU and leading APs, and between the leading APs and APs, is proposed so as to realize the dynamic AP clustering and achieve high end-to-end rates. To facilitate the proposed scheme, MDD is leveraged to activate the full-duplex (FD) operation among two wireless fronthaul links over orthogonal subcarrier sets \cite{li2021multicarrier}. Compared with its  time-division duplex (TDD) and in-band full-duplex (IBFD) counterparts, MDD is able to largely exploit the time-frequency resources at the little expense of self-interference (SI) and cross-layer interference (CLI) in CF-mMIMO systems \cite{li2022spectral}. The term CLI refers to the interference caused by the node's transmitted signal onto its neighboring nodes during signal reception at the same frequency band and time slot.  On the other hand, with the increasing demands for higher rates and better quality-of-service (QoS) in the industrial networks, THz technique could provide over Gigabit-per-second (Gbps) rates with abundant unoccupied bandwidth \cite{lin2015adaptive}, and achieve the high synchronization accuracy between the CPU and all the APs \cite{demirhan2022enabling}. Hence, THz-based fronthaul is expected to fulfill the goals of replacing the conventional wired fronthaul. In addition, the combination of THz and MDD has the synergy to mitigate the effect of SI and CLI. More specifically, in the systems operated in the FD mode, the interference cancellation depends on the propagation-, analog-, and digital-domain approaches \cite{li2017full}. According to the principle of MDD, digital-domain interference can be thoroughly canceled by the built-in analog-to-digital converter (ADC) of receiver chains with the FFT operation \cite{li2021multicarrier}. Therefore, with the aid of extremely large path loss in the propagation domain and the FFT operation in the digital domain, MDD-based systems operated over the THz band are expected to be completely free from SI and CLI.

\subsection{Contributions}      
The main contributions of this paper are as follows:
\begin{itemize}
\item In order to realize the fully-wireless fronthaul for indoor industrial networks, an MDD-enabled two-tier THz scheme is proposed, which is capable of providing high data rates and wide converge for the distributed APs. Additionally, the proposed scheme can significantly save the cost of wired fronthaul cables and dynamically implement AP clustering to adapt to the variations in the network with different numbers of APs and devices.
\item The overall end-to-end optimization considering two THz fronthaul and sub-6 GHz access links is formulated as a max-min fairness problem. The problem jointly considers the AP clustering, device selection, the assignment of subcarrier sets between two fronthaul links, power and subcarrier allocation at the CPU and CAPs, and power allocation at APs. To deal with the complicated problem, heuristic methods are first leveraged to relax the binary variables with regard to AP clustering, device selction and the assignment of subcarrier sets. Then, the original problem is transformed into two individual fronthaul and access max-min sub-problems. The overall end-to-end rate maximization is solved by iteratively optimizing the assignment of subcarrier sets and the number of AP clusters. 
\item To further facilitate the proposed fronthaul scheme, an advanced MDD frame structure enabling three parallel data streams is proposed. For comparison, the frame structure in TDD-based two-tier THz fronthaul scheme is also specifically designed for the end-to-end rate maximization.
\item To verify the performance of the proposed scheme, we apply it to a practical indoor industrial scenario relying on the ray-tracing (RT) channel measurements. Simulation results show that MDD outperforms TDD in the two-tier fully-wireless fronthaul scheme. Moreover, when provided with sufficient THz bandwidth, indoor industrial CF-mMIMO systems can be free from the wired fronthaul with the aid of our proposed MDD-enabled two-tier THz fronthaul scheme.
\end{itemize}

\setlength{\abovedisplayskip}{3pt}
\setlength{\belowdisplayskip}{3pt}

\section{System Model}\label{section:MDDTHz:SM}
We study the downlink (DL) transmission in an indoor industrial cell-free massive MIMO (CF-mMIMO) network, where the CPU is equipped with a $\left(N_{\text{CPU}}^{'} \times N_{\text{CPU}}^{''}\right)$ uniform planar array (UPA) having $N_{\text{CPU}}=N_{\text{CPU}}^{'} N_{\text{CPU}}^{''}$ antennas, and sends data to APs via wireless fronthaul channels over THz frequency-band, while $Q$ APs, thereafter, transmit the received data to $U$ single-antenna devices via access channels over sub-6 GHz band. To enable such a dual-band architecture, every AP employs both the high- and low-frequency RF chains \cite{gao2015mmwave,demirhan2022enabling}, and each RF chain is connected with a $\left(N_{\text{AP}}^{'} \times N_{\text{AP}}^{''}\right)$ UPA of $N_{\text{AP}}=N_{\text{AP}}^{'}N_{\text{AP}}^{''}$ antennas at transmitter and single antenna at receiver\footnote{In this paper, we leverage two separate antenna arrays at AP to realize simultaneous signal transmission over the THz and sub-6 GHz bands, respectively. Furthermore, considering the system's overhead and space usage, an integrated low-band and high-band MIMO antenna module proposed in \cite{hussain2022integrated} can be applied to make our dual-band architecture more efficient.}. 

In general, it is straightforward to connect every AP with CPU through wireless fronthaul links \cite{masoumi2019performance}. However, as APs are densely distributed in CF-mMIMO systems, it would be a demanding task to compute all the beamforming matrices at CPU, while the management of inter-beam interference is also highly challenging. To circumvent this problem, we propose a two-tier fully-wireless fronthaul architecture, as shown in Fig. \ref{figure-MDDTHz-Archi}, where $Q$ APs are classified into $L$ distinct clusters, expressed as $\mathcal{C}_l$, $\forall l\in \left\{1,...,L\right\}$, and $\left|\bigcup_{l=1}^L \mathcal{C}_l\right|=Q$. The CPU only establishes fronthaul link with a leading AP, namely computational AP (CAP), in each cluster. The CAP of each cluster takes charge of both the fronthaul transmissions with the remaining APs and the signal processing tasks, such as DL precoding and channel estimation\footnote{In order to achieve the dynamic AP clustering and CAP selection as discussed in Section \ref{section:MDDTHz:APC}, each AP is equipped with a basic baseband processor such that any AP within a cluster can be chosen as the CAP. This kind of configuration, similar to that in the distributed CF-mMIMO \cite{demir2021foundations}, inevitably increases the overhead of system, but it can significantly reduce the cost on fiber-optic cables, and its budget may be further reduced by the developed hardware techniques in the future.}. This architecture enables CF-mMIMO to get rid of the wired fronthaul connections, thereby exhibiting high flexibility and scalability. Furthermore, it conforms to the principle of the user-centric CF-mMIMO, where CAPs act as geographically distributed edge-cloud processors to divide the tasks of CPU\footnote{Although the CAPs employed with low-frequency RF chains can implement DL transmissions over sub-6 GHz band like other APs, it is assumed to only transmit and receive fronthaul signal for the sake of saving computational overhead.}. Moreover, we assume that all the APs in one cluster serve the same device, and the group of APs serving device $u$, expressed as $\mathcal{G}_u$, is determined via selecting the set of clusters, such that $\mathcal{G}_u \subseteq \left\{1,..,L\right\}$. Let $\mathcal{U}_l$ denote the set of devices concurrently served by cluster $l$. We further assume that each cluster can serve at most $U_{\text{max}}$ devices, i.e., $\left|\mathcal{U}_l\right|\leq U_{\text{max}}, \forall l \in\left\{1,...,L\right\}$, while each device can choose up to $C_{\text{max}}$ serving clusters, meaning that $\left|\mathcal{G}_u\right|\leq C_{\text{max}}, \forall u \in\left\{1,...,U\right\}$.


\begin{figure}
\centering
\includegraphics[width=0.5\linewidth]{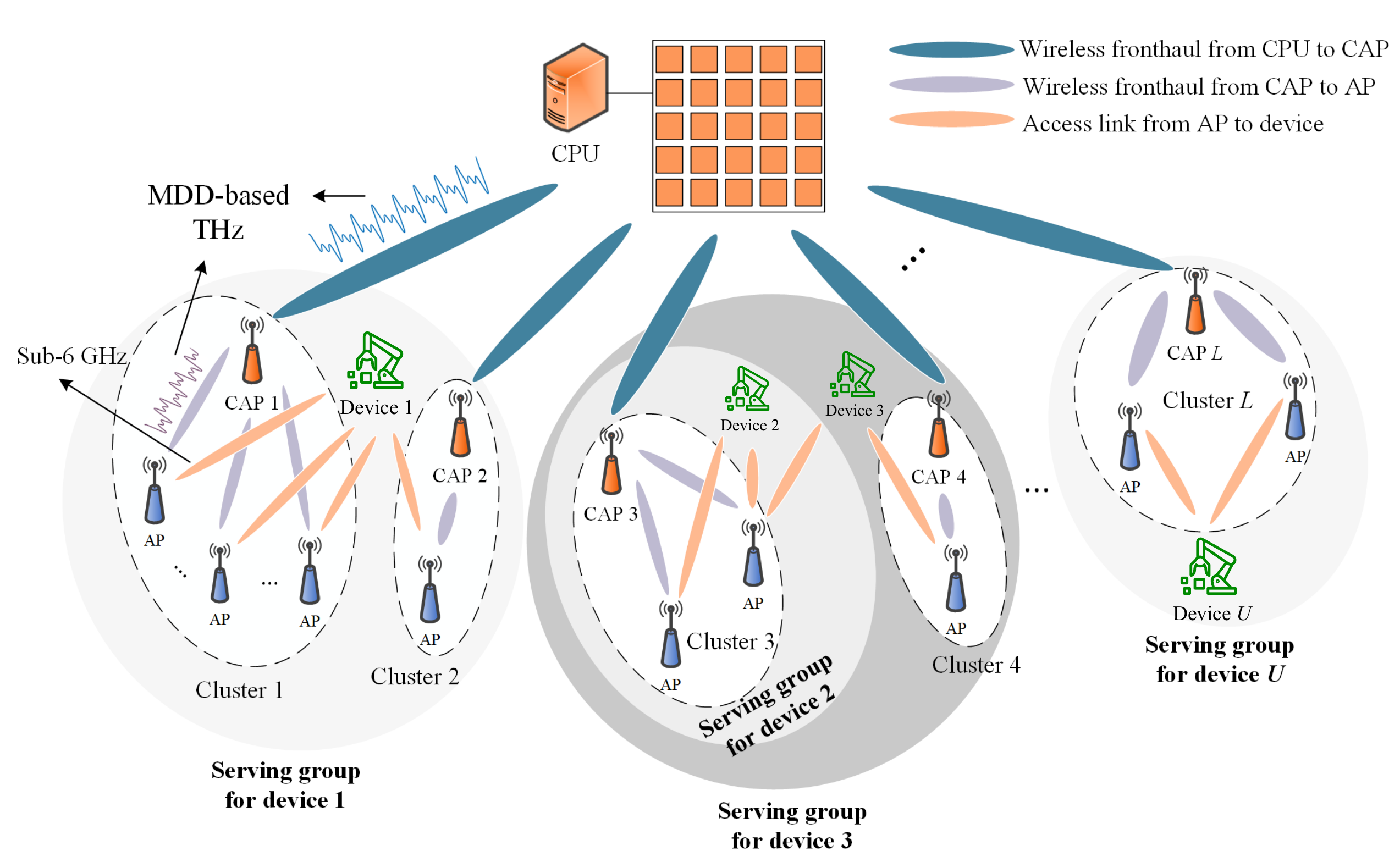}
\vspace{-0.5cm}
\caption{The proposed MDD-Enabled THz Fronthaul in indoor industrial CF-mMIMO.}
\label{figure-MDDTHz-Archi}\vspace{-0.7cm}
\end{figure}

To overcome multi-path fading, especially in wideband THz communications, we resort to OFDM for signal transmission. In particular, we assume that CPU and CAPs transmit MDD-based fronthaul signals relying on two mutually orthogonal sets of subcarriers in the same THz band, i.e., $\mathcal{M}_{\text{CC}}$ and $\mathcal{M}_{\text{CA}}$, such that $\mathcal{M}_{\text{CC}}\bigcap \mathcal{M}_{\text{CA}}=\emptyset$ and $\mathcal{M}_{\text{CC}} \bigcup \mathcal{M}_{\text{CA}}= \mathcal{M}_{\text{FH}}$, while the data transmitted from APs to devices are modulated on the subcarrier set $\mathcal{M}_{\text{AD}}$ over the sub-6 GHz band. In CF-mMIMO, it is intuitive that each cluster of APs may simultaneously serve more than one device, which requires CAPs to receive (transmit) the multi-device signals from (to) CPU (APs). To this end, in the proposed multicarrer system, the CPU and CAPs adopt precoders to map the different devices' signals within each cluster onto different subcarriers of $\mathcal{M}_{\text{CC}}$ and $\mathcal{M}_{\text{CA}}$, respectively.



\subsection{Fronthaul Transmissions}
The fronthaul in our proposed architecture consists of CPU-to-CAP and CAP-to-AP links. As these links rely on THz communications, we first establish the delay-$t$ THz channel model between the $\left(N^{'}\times N^{''}\right)$ UPA transmitter and the single-antenna receiver as \cite{yuan2020hybrid}
\begin{align}\label{eq:MDD-THz:THzChannelt}
&\pmb{h}_t(f_m, d) = \pmb{h}_t^{\text{LoS}}(f_m, d)+\pmb{h}_t^{\text{NLoS}}(f_m, d)= \sqrt{N^{'}N^{''}} \alpha^{\text{LoS}}(f_m, d)p_{\text{rc}}(tT_{\text{s}}-\delta) G_{\text{tx}} G_{\text{rx}} \pmb{a}_{\text{tx}}\left(\phi^{\text{tx}}, \theta^{\text{tx}}\right) \nonumber \\
&+\sqrt{N^{'}N^{''}/N_{\text{ray}}}  \sum_{i=1}^{N_{\text{ray}}} \alpha_i(f_m, d) G_{\text{tx}} G_{\text{rx}} p_{\text{rc}}(tT_{\text{s}}-\delta_i) \pmb{a}_{\text{tx}}\left(\phi_i^{\text{tx}}, \theta_i^{\text{tx}}\right),
\end{align}
where $f_m$ is the center frequency of the $m$-th subcarrier, $d$ is the path distance, $G_{\text{tx}}$ and $G_{\text{rx}}$ denote the antenna gains of transmitter and receiver, respectively, and $p_{\text{rc}}(\cdot)$ denotes the combined effect of pulse shaping and analog filtering, with $T_{\text{s}}$ and $\delta$ denoting the sampling interval and delay. Moreover, $\alpha^{\text{LoS}}(f_m, d)$, dependent on frequency and path distance, is the complex gain of the LoS path component, and its squared norm is given by
$\left|\alpha^{\text{LoS}}(f_m, d)\right|^2=L_{\text{spread}}(f_m,d)L_{\text{abs}}(f_m,d),$
where $L_{\text{spread}}(f_m,d)=\left(\frac{c}{4 \pi f_m d}\right)^2$ and $L_{\text{abs}}(f_m,d)=e^{-k_{\text{abs}}(f_m)d}$ with $k_{\text{abs}}(f_m)$ denoting the frequency-related absorption coefficient, which is determined by the propagation medium \cite{lin2015adaptive}. By contrast, the complex gain of the NLoS path is subject to not only the distance and frequency, but also the reflection coefficient, which is given by $\left|\alpha_i(f_m, d)\right|^2=\Gamma_i^2(f_m)\left|\alpha^{\text{LoS}}(f_m, d)\right|^2$, where $\Gamma_i^2(f_m)$ is the product of the Fresnel reflection coefficient and the Rayleigh roughness factor \cite{lin2015adaptive}. Furthermore, $\pmb{a}_{\text{tx}}(\phi^{\text{tx}},\theta^{\text{tx}})$ denotes the UPA response vector at the transmitter \cite{lin2015adaptive},
where $\phi$, $\theta$ denote azimuth angle of departure (AoD) and elevation AoD, receptively.


Given the delay-$t$ THz channel model in \eqref{eq:MDD-THz:THzChannelt}, the $m$-th subcarrier channel can be expressed as\footnote{Note that as the main objectives of this paper are to design a two-tier fully-wireless fronthaul architecture and jointly implement resource allocation among fronthaul and access links, the problems of beam squint and carrier frequency offset arising from THz channels have not been considered \cite{yuan2020hybrid,dovelos2021channel}, which are left to our future study.}
\begin{equation}\label{eq:MDD-THz:THzChannelf}
\pmb{h}[m] = \sum_{t=0}^{T_{\text{delay}}-1}\pmb{h}_t(f_m, d)e^{-j\frac{2\pi mt}{\left|\mathcal{M}_{\text{FH}}\right|}}. 
\end{equation}

\subsubsection{CPU-to-CAP Link}
The CPU beamforms the signal toward CAPs over the subcarriers selected from $\mathcal{M}_{\text{CC}}$. Let $x_{l,u}[m]$ denote the data intended for device $u$ associated with the $l$-th cluster over the $m$-th subcarrier, with $\mathbb{E}\left[\left|x_{l,u}[m]\right|^2\right]=1$. Then, the $m$-th subcarrier signal received at the $l$-th CAP can be expressed as
\begin{small}
\begin{align}\label{eq:MDD-THz:Rylm}
&y_{l}[m] = \underbrace{\pmb{h}_l^H[m]\sum_{u \in \mathcal{U}_l} \gamma_{l,u,m}\sqrt{p_{l,u}[m]}\pmb{f}_{l,u}[m]x_{l,u}[m]}_{\text{Desired signal}}+\underbrace{\pmb{h}_l^H[m]\sum_{ l^{\prime}\neq l }\sum_{u^{\prime} \in \mathcal{U}_{l^{\prime}}} \gamma_{l^{\prime},u^{\prime},m}\sqrt{p_{l^{\prime},u^{\prime}}[m]}\pmb{f}_{l^{\prime},u^{\prime}}[m]x_{l^{\prime},u^{\prime}}[m]}_{\text{Inter-CAP interference}} \nonumber \\
&+ \text{SI}_l +  n_l[m], \ \ \forall m \in \mathcal{M}_{\text{CC}},
\end{align}
\end{small}
where 
$\gamma_{l,u,m}=1$ denotes that the CPU sends the data of device $u$ to the $l$-th CAP over the $m$-th subcarrier of $\mathcal{M}_{\text{CC}}$, and $\gamma_{l,u,m}=0$, otherwise, $p_{l,u}[m]$ is the power coefficient with respect to the device $u$ within the $l$-th cluster over subcarrier $m$, $\pmb{h}_l[m] \in \mathbb{C}^{N_{\text{CPU}} \times 1}$ is the channel impulse response of fronthaul channel from CPU to CAP $l$, which follows the expression of \eqref{eq:MDD-THz:THzChannelf}, $\pmb{f}_{l,u}[m] \in \mathbb{C}^{{N}_{\text{CPU}} \times 1}$ is the precoder specific for device $u$ associated with the $l$-th cluster over subcarrier $m$, 
$n_l[m] \sim \mathcal{CN}(0,\sigma_{\text{l}}^2)$ is the additive white Gaussian noise. Note that, as CAP leverages the subcarrier diversity to simultaneously receive the data of different devices, one subcarrier cannot be assigned to multiple devices within one cluster, i.e., $\gamma_{l,u_1,m}\gamma_{l,u_2,m}=0$ when $u_1\neq u_2$. Furthermore, $\text{SI}_l\sim \mathcal{CN}(0,\sigma_{\text{SI}}^2)$ denotes the residual SI caused by the CAP's own transmitter.

\begin{remark}
Compared with the conventional in-band full-duplex, MDD enables simultaneous DL and UL transmissions over orthogonal subcarrier sets in the same band. Owing to this feature, the SI is orthogonal to the UL signal in the digital domain at the CAP's receiver, and hence can be readily canceled with the FFT operation. However, as SI is strongly coupled with the UL signal in the analog domain, the analog-domain SIC is still indispensable. Otherwise, the interference of overlarge power may overwhelm the desired signal of relatively small power, leading to the severe quantization noise after its passing through ADC. It is noteworthy that the suppression of the analog-domain SI in high-frequency bands is much easier than that of the digital-domain SI, due to the highly directional beam and large path-loss \cite{yadav2018full}. To this end, the SI mitigation in MDD-based THz systems can be achieved at less expenses. Hence, without deviation of the core of this paper, the detailed discussion of SI cancellation is omitted. Instead, we model the residual interference as Gaussian noise consisting of the combined effect of the additive noise introduced by automatic gain control, non-linearity of ADC and the phase noise arose from oscillator due to RF imperfection \cite{day2012full2}. 
\end{remark}

Based on \eqref{eq:MDD-THz:Rylm}, the achievable rate of CAP $l$ corresponding to device $u$ can be expressed as
\begin{equation}\label{eq:MDD-THz:CCSINR}
\small
C_{l,u}^{\text{CC}} = b^{\text{FH}}\sum_{m \in \mathcal{M}_{\text{CC}}}\log\left(1+\frac{\gamma_{l,u,m}p_{l,u}[m]\left|\pmb{h}_l^H[m]\pmb{f}_{l,u}[m]\right|^2}{\sum_{ l^{\prime}\neq l }\sum_{u^{\prime} \in \mathcal{U}_{l^{\prime}}} \gamma_{l^{\prime},u^{\prime},m}p_{l^{\prime},u^{\prime}}[m]\left|\pmb{h}_l^H[m]\pmb{f}_{l^{\prime},u^{\prime}}[m]\right|^2+\sigma_{\text{SI}}^2+\sigma_{\text{l}}^2}\right),
\end{equation}
where $b^{\text{FH}}$ denotes the subcarrier bandwidth in THz band. Since device $u$ may have $\left|\mathcal{G}_u\right|$ serving clusters, its message will be received by multiple CAPs. In this case, the effective rate of device $u$ in the CPU-to-CAP links can be defined as the minimum rate of the CAPs in the serving cluster $\mathcal{G}_u$, which can be written as
$C_{u}^{\text{CC}} = \min_{l \in \mathcal{G}_u} \ \left\{C_{l,u}^{\text{CC}}\right\}$.


\subsubsection{CAP-to-AP Link}
Each CAP $l$ forwards the CPU signal, relying on $\mathcal{M}_{\text{CA}}$ subcarriers, to the remaining APs within its serving cluster $\mathcal{C}_l$. Let $\tilde{x}_{l,q,u}[\bar{m}]$ denote the data intended for device $u$ associated with the $l$-th cluster and $q$-th AP over the $\bar{m}$-th subcarrier. Then, the signal to be transmitted at the CAP $l$ can be expressed as
$\pmb{s}_l = \sum_{\bar{m} \in \mathcal{M}_{\text{CA}}} \pmb{s}_l[\bar{m}]$,
where $\pmb{s}_l[\bar{m}]=\sum_{q \in \mathcal{C}_l}\sum_{u \in \mathcal{U}_l}\gamma_{l,q,u,\bar{m}}\sqrt{p_{l,q,u}[\bar{m}]}\pmb{w}_{l,q,u}[\bar{m}]\tilde{x}_{l,q,u}[\bar{m}]$, $\gamma_{l,q,u,\bar{m}}=1$ denotes that the $l$-th CAP sends the data of device $u$ to the $q$-th AP over the $\bar{m}$-th subcarrier in $\mathcal{M}_{\text{CA}}$, and $\gamma_{l,q,u,\bar{m}}=0$, otherwise, $p_{l,q,u}[\bar{m}]$ is the power coefficient associated with the device $u$ at AP $q$ over subcarrier $\bar{m}$, $\pmb{w}_{l,q,u}[\bar{m}] \in \mathbb{C}^{{N}_{\text{AP}} \times 1}$ is the precoder specific for device $u$ associated with AP $q$ over subcarrier $\bar{m}$, 
 $n_l[\bar{m}] \sim \mathcal{CN}(0,\sigma_{\text{n}}^2)$ is the additive white Gaussian noise. Analogously, as AP only has one receiver antenna, the receiving of different devices' data depends on subcarrier diversity, leading to $\gamma_{l,q,u_1,\bar{m}}\gamma_{l,q,u_2,\bar{m}}=0$, when $u_1\neq u_2$.

\vspace{0.1cm}
Then, the signal received at the $q$-th AP of $\mathcal{C}_l$ over subcarrier $\bar{m}$ can be written as
\begin{small}
\begin{align}\label{eq:MDD-THz:Rylq}
&y_{l,q}[\bar{m}]= \underbrace{\pmb{h}_{l,q}^H\!\sum_{u \in \mathcal{U}_l}\!\gamma_{l,q,u,\bar{m}}\sqrt{p_{l,q,u}[\bar{m}]}\pmb{w}_{l,q,u}[\bar{m}]\tilde{x}_{l,q,u}[\bar{m}]}_{\text{Desired signal}} \! +\! \underbrace{\pmb{h}_{l,q}^H\!\sum_{q^{\prime}\neq q}\!\sum_{u^{\prime} \in \mathcal{U}_l}\!\gamma_{l,q^{\prime},u^{\prime},\bar{m}}\sqrt{p_{l,q^{\prime},u^{\prime}}[\bar{m}]}\pmb{w}_{l,q^{\prime},u^{\prime}}[\bar{m}]\tilde{x}_{l,q^{\prime},u^{\prime}}[\bar{m}]}_{\text{Intra-cluster interference}} \nonumber \\
&+ \underbrace{\sum_{l^{\prime}\neq l}  \pmb{h}_{l^{\prime},q}^H\pmb{s}_{l^{\prime}}[\bar{m}]}_{\text{Inter-cluster interference}} + n_{l,q}[\bar{m}], \ \ \forall \bar{m} \in \mathcal{M}_{\text{CA}},
\end{align}
\end{small}
where $\pmb{h}_{l,q}[\bar{m}] \in \mathbb{C}^{N_{\text{AP}} \times 1}$ is the fronthaul channel between CAP $l$ and AP $q$ following the expression of \eqref{eq:MDD-THz:THzChannelf}.

\begin{remark}
Similar to the SI at CAP receiver, as two fronhaul links operate over orthogonal subcarriers, the cross-link interference (CLI) imposed by CPU at each AP receiver is orthogonal to the CAP-to-AP fronthaul signal in the digital domain. Moreover, due to the employment of beamforming at CPU and extremely high large-scale fading between CPU and APs, the analog-domain CLI significantly deteriorates before arriving at ADC of APs' receiver. Accordingly, the CLI can be neglected in this case. 
\end{remark}

Consequently, based on \eqref{eq:MDD-THz:Rylq}, the achievable rate of AP $q$ in the cluster $\mathcal{C}_l$ corresponding to device $u$ can be expressed as
\begin{small}
\begin{align}\label{eq:MDD-THz:CASINR}
&C_{l,q,u}^{\text{CA}} =b^{\text{FH}} \cdot \nonumber \\
&{\scriptsize \sum_{m \in \mathcal{M}_{\text{CC}}}\!\log\!\left(\!1\!+\!\frac{\gamma_{l,q,u,\bar{m}}p_{l,q,u}[\bar{m}]\!\left|\pmb{h}_{l,q}^H[\bar{m}]\pmb{w}_{l,q,u}[\bar{m}]\right|^2}{\sum_{ q^{\prime}\neq q }\!\sum_{u^{\prime} \in \mathcal{U}_{l}}\! \gamma_{l,q^{\prime},u^{\prime},\bar{m}}p_{l,q^{\prime},u^{\prime}}[\bar{m}]\!\left|\pmb{h}_{l,q}^H[\bar{m}]\pmb{w}_{l,q^{\prime},u^{\prime}}[\bar{m}]\right|^2\!+\! \sum_{l^{\prime}\neq l}\!\left|\mathbb{E}\left\{\pmb{h}_{l^{\prime},q}^H  \pmb{s}_{l^{\prime}}[\bar{m}]\right\}\right|^2\!+\!\sigma_{q}^2}\right).}
\end{align} 
\end{small}
Therefore,
the effective rate in cluster $\mathcal{C}_l$ corresponding to device $u$ can be written as
$C_{u}^{\text{CA}} = \min_{l \in \mathcal{G}_u, q \in \mathcal{C}_l} \ \left\{C_{l,q,u}^{\text{CA}}\right\}$.

\subsection{Device Communications via AP-to-Device Link}
Once the APs decode the fronthaul signal received from the CAP, it will communicate with devices via the sub-6 GHz access channels, which are assumed to experience flat fading with $|\mathcal{M}_{\text{AD}}|=1$, for ease of discussion. In particular, the channel between the $q$-th AP in $\mathcal{C}_l$ and device $u$, i.e., $\tilde{\pmb{h}}_{l,q,u} \in  \mathbb{C}^{N_{\text{AP}} \times 1}$, is modeled as $(\tilde{\pmb{h}}_{l,q,u})_i=\sqrt{\beta_{l,q,u}} \ \alpha_{\text{s}}$ with $\beta_{l,q,u}$ and $\alpha_{\text{s}}$ denoting the large and small-scale fading coefficients, respectively. 
 
The transmitted signal from the $q$-th AP within cluster $\mathcal{C}_l$ can be written as
$\pmb{s}_{l,q} = \sum_{u \in \mathcal{U}_l}\sqrt{p_{l,q,u}}\pmb{v}_{l,q,u} x_{u}$,
where $p_{l,q,u}$ denotes the transmit power coefficient for device $u$ at the $q$-th AP, $\pmb{v}_{l,q,u} \in \mathbb{C}^{{N}_{\text{AP}} \times 1}$ is the precoder specific for device $u$ at the $q$-th AP. 
Then, the received signal at device $u$ can be expressed as
\begin{equation}
\small
y_u = \underbrace{\sum_{l \in \mathcal{G}_u}\pmb{h}_{l,u}^H\pmb{P}_{l,u}\pmb{v}_{l,u}x_{u}}_{\text{Desired signal}}+\underbrace{\sum_{l \in \mathcal{G}_u}\sum_{u^{\prime} \in \mathcal{U}_l,u^{\prime}\neq u}\pmb{h}_{l,u}^H\pmb{P}_{l,u^{\prime}} \pmb{v}_{l,u^{\prime}}x_{u^{\prime}}}_{\text{Inter-device interference}} 
+ \underbrace{\sum_{l^{\prime} \notin \mathcal{G}_u} \sum_{q \in \mathcal{C}_{l^{\prime}}} \tilde{\pmb{h}}_{l^{\prime}q,u}^H \pmb{s}_{l^{\prime},q}}_{\text{Inter-cluster interference}} + n_u,
\end{equation}
where $\small \pmb{h}_{l,u}\!=\!\left[\tilde{\pmb{h}}_{l,1,u}^T,...,\tilde{\pmb{h}}_{l,\left|\mathcal{C}_l\right|,u}^T\right]^T$\!, \!$\small \tilde{\pmb{P}}_{l,u}\!=\!\left(\pmb{P}_{l,u}\!\otimes \!\pmb{I}_{N_\text{AP}}\right)$ with $\small \pmb{P}_{l,u}\!=\!\text{diag}\left(\left[\sqrt{p_{l,1,u}},...,\sqrt{p_{l,{\left|\mathcal{C}_l\right|},u}}\right]\right)$, and $\small \pmb{v}_{l,u}=\left[\left(\pmb{v}_{l,1,u}\right)^T,...,\left(\pmb{v}_{l,\left|\mathcal{C}_l\right|,u}\right)^T\right]^T$.
The achievable rate of the AP-to-Device link corresponding to device $u$ can be written as
\begin{equation}
\small
{\scriptsize C_u^{\text{AD}}\!=\!B^{\text{AD}}\log\!\left(\!1+\frac{\left|\sum_{l \in \mathcal{G}_u}\pmb{h}_{l,u}^H\tilde{\pmb{P}}_{l,u} \pmb{v}_{l,u}\right|^2}{\sum_{l \in \mathcal{G}_u}\sum_{u^{\prime} \in \mathcal{U}_l,u^{\prime}\neq u}\left|\pmb{h}_{l,u}^H\tilde{\pmb{P}}_{l,u^{\prime}}\pmb{v}_{l,u^{\prime}}\right|^2+\sum_{l^{\prime} \notin \mathcal{G}_u} \sum_{q \in \mathcal{C}_{l^{\prime}}}\left|\mathbb{E}\left\{\tilde{\pmb{h}}_{l^{\prime}q,u}^H \pmb{s}_{l^{\prime},q}\right\}\right|^2+\sigma^2_u}\!\right),}
\end{equation}
where $B^{\text{AD}}$ denotes the bandwidth of the sub-6 GHz band.

\subsection{Optimization Problem}
Considering the fact that the data flows from CPU to each device, as shown in Fig. \ref{figure-MDDTHz-Archi}, the achievable rate of device $u$ in $bits/s$ can be written as
$C_u = \text{min}\left\{C_u^{\text{CC}},C_u^{\text{CA}},C_u^{\text{AD}}\right\}$.
Then, to guarantee the fairness among all the devices, the optimization problem can be stated as:
\vspace{-0.3cm}
\begin{subequations}
\label{eq:MDD-THz:opt}
\begin{align}
&\underset{\mathcal{C},\mathcal{G},\mathcal{M},\left\{\pmb{f}\right\},\left\{\pmb{w}\right\},\left\{\pmb{v}\right\},\left\{\gamma\right\}\left\{p\right\}}{\text{maximize}} \ \min_{u\in\left\{1,...,U\right\}}C_u \\
\text{s.t.} \ \ 
& \sum_{l=1}^L \left|\mathcal{C}_l\right|=Q, \ \mathcal{C}_{l_1} \cap \mathcal{C}_{l_2} = \emptyset, \text{if} \  l_1 \neq l_2, \\
& \left|\mathcal{U}_l\right|\leq U_{\text{max}}, \forall l \in\left\{1,...,L\right\}, \\
& \left|\mathcal{G}_u\right| \leq C_{\text{max}}, \forall u \in \left\{1,..,U\right\},\\
& \mathcal{M}_{\text{CC}}\cup \mathcal{M}_{\text{CA}} = \mathcal{M}_{\text{FH}}, \mathcal{M}_{\text{CC}} \cap \mathcal{M}_{\text{CA}} =\emptyset, \\
&\gamma_{l,u,m}\in \left\{0,1\right\}, \ \gamma_{l,u_1,m}\gamma_{l,u_2,m}=0, \forall u_1\neq u_2, \\
&\gamma_{l,q,u,\bar{m}}\in \left\{0,1\right\}, \ \gamma_{l,q,u_1,\bar{m}}\gamma_{l,q,u_2,\bar{m}}=0, \forall u_1\neq u_2, \\
& \sum_{m \in \mathcal{M}_{\text{CC}}}\sum_{l=1}^L \sum_{u=1}^U \gamma_{l,u,m}p_{l,u}[m]\left\|\pmb{f}_{l,u}[m]\right\|_2^2 \leq P_{\text{CPU}}, \\
& \sum_{\bar{m} \in \mathcal{M}_{\text{CA}}}\sum_{q \in \mathcal{C}_{l}}\sum_{u \in \mathcal{U}_{l}} \gamma_{l,q,u,\bar{m}}p_{l,q,u}[\bar{m}]\left\|\pmb{w}_{q,u}[\bar{m}]\right\|_2^2 \leq P_{\text{AP}}, \forall l \in \left\{1,...,L\right\}, \\
& \sum_{u \in \mathcal{U}_l}p_{l,q,u}\left\|\pmb{v}_{l,q,u}\right\|_2^2 \leq P_{\text{AP}}, \forall l \in \left\{1,...,L\right\}, q \in \left\{1,...,Q\right\}, 
\end{align} 
\end{subequations}
It can be observed that the overall optimization problem is non-convex and involved with different types of variables. Hence, it is extremely time-consuming and high-complexity to find an optimal solution. To work around this issue, in the next section, we aim to relieve the constraints (\ref{eq:MDD-THz:opt}b)-(\ref{eq:MDD-THz:opt}e) so as to find the approximate solution with high computational efficiency.


\section{Relaxation of Optimization Problem} \label{section:MDDTHz:relax}
According to formulation \eqref{eq:MDD-THz:opt}, the AP clustering, device selection and the assignment of subcarrier sets consist of a large number of binary variables. It is impractical to simultaneously optimize all the binary variables, especially for the networks of large size. Therefore, we will leverage the heuristic but efficient algorithms to deal with these problems, provided that the number of clusters $L$ and the sizes of $\mathcal{M}_{\text{CC}}$ and $\mathcal{M}_{\text{CA}}$ are already determined.
 
\subsection{AP Clustering}\label{section:MDDTHz:APC}
In the light of the fact that the indoor THz-based communications are distance-sensitive due to the extremely high path loss, it is reasonable to implement AP clustering mainly based on the proximity principle. In other words, a group of APs situated closely to each other are most likely to be arranged in a cluster. In doing so, the high-rate fronthaul transmissions between CAP and APs within the cluster can be achieved, and the centralized beamforming design can produce more focused beams toward devices. Apart from AP clustering, the selection of CAP in a cluster is critical to the proposed cluster-based wireless fronthaul scheme. Specifically, each CAP acts as the pivot of two fronthaul links, simultaneously receiving data from CPU and transmitting data to the APs within the cluster. Hence, the potential CAP candidates should have the capability to build the reliable wireless fronthaul connections with both CPU and APs. Additionally, as THz communications are leveraged for fronthaul, the selection of CAP should consider mainly the large-scale fading between CPU and APs. In this paper, three AP clustering approaches, namely Dynamic clustering (DC), Static clustering (SC) and integrated DC and SC (IDSC), are taken into account and presented as follows: 
\begin{enumerate}[(i)]
\item DC: The DC algorithm consists of two steps. In the first step, the AP clustering can be implemented by {\em K-medoids} method \cite{park2009simple}, and $\left\{\mathcal{C}_l\right\}_{l=1}^L$ is then obtained.
In the second step, the CAP of each cluster needs to be determined. Following \eqref{eq:MDD-THz:THzChannelt}, let $\alpha^{\text{CC-LoS}}_q$ denote the LoS complex gain between CPU and the $q$-th AP, while $\alpha^{\text{CA-LoS}}_{q,q^{'}}$ denote the LoS complex gain between the $q$-th AP and the $q^{'}$-th AP in the same cluster. Then, the total LoS channel gain corresponding to AP $q$ is calculated as $\Psi_q=|\alpha^{\text{CC-LoS}}_q|+\sum_{q^{'} \in \mathcal{C}_l, q^{'}\neq q}|\alpha^{\text{CA-LoS}}_{q,q^{'}}|$, and
the CAP in cluster $l$ can be selected as $\text{CAP}_l=\argmax_{q}\Psi_q$. 
\item SC: As we can see from DC that the complexity of {\em K-medoids} increases with the number of APs. Moreover, the selected CAPs of different AP clusters may be close to each other, leading to high inter-cluster interference. Hence, from the perspective of practical implementation, in SC, we can simply divide the whole communication area into $L$ sub-areas, and the APs located within a sub-area are classified as one AP cluster. Instead of dynamically selecting CAP, each sub-area manually deploys one CAP at its central point.
\item IDSC: Followed by DC and SC methods, it is intuitive to propose an integrated DC and SC method. Specifically, in IDSC, the AP clustering is achieved in the same way as the SC method, while the CAP selection follows the approach in the DC method.
\end{enumerate}

\begin{remark}
As $L$ increases, some AP clusters may only have a single AP, i.e., CAP. In this case, it is assumed that the CAP receives signal via CPU-to-CAP fronthaul, and then directly communicates with devices. Furthermore, when $L$ is increased to $Q$, i.e., $L=Q$, the proposed MDD-based wireless fronthaul degrades to the conventional single-tier fronthaul network, which can be deemed as a benchmark for comparison in Section \ref{sec:MDDTHz:sim}.
\end{remark}

\subsection{Device Selection}
The device selection here is also known as the user-centric clustering \cite{buzzi2017cell,bjornson2020scalable}. After AP clustering, since all the APs within a cluster are in vicinity to each other and also receive the fronthaul signals from the same CAP, we assume that all the APs in one cluster serve the same set of devices. In other words, in our proposed user-centric CF-mMIMO, devices actively select AP clusters rather than individual APs. 
To implement the device selection, let $v_{l,u}=\sum_{q \in \mathcal{C}_l}\left\|\tilde{\pmb{h}}_{l,q,u}\right\|_2^2$ denote the channel gain between the $l$-th cluster and $u$-th device.
The selection process consists of several steps: (i) To begin with, we set $\mathcal{G}_{u}=\mathcal{U}_l=\emptyset$ and $\mathcal{L}=\left\{1,...,L\right\}$, $\forall u, l$, and sequentially assign an initial cluster to each device $u$, following the rule of $l^{\ast}=\argmax_{l\in \mathcal{L}}{v_{l,u}}$. Then, we update $\mathcal{G}_{u}=\mathcal{G}_{u} \bigcup \left\{l^{\ast}\right\}$ and $\mathcal{U}_l = \mathcal{U}_l \bigcup \left\{u\right\}$. During this step, if $\left|\mathcal{U}_l\right|>U_{\text{max}}$, the $l$-th cluster can no longer be chosen, and $\mathcal{L}=\mathcal{L} \backslash  l$; (ii) Considering the fairness among devices, we re-rank the devices in the order of total channel gain, i.e., $v_{u}=\sum_{l\in \mathcal{G}_u}\sum_{q \in \mathcal{C}_l}\left\|\tilde{\pmb{h}}_{l,q,u}\right\|_2^2$. The device having the smallest total channel gain is allowed to first choose the optimal $l^{\ast}$ from the remaining clusters. The corresponding sets are then updated; (iii) Step (ii) is repeated until $\left|\mathcal{U}_l\right|=U_{\text{max}}$ or $\left|\mathcal{G}_{u}\right| = C_{\text{max}}$ for $\forall l \in \left\{1,...,L\right\}$ and $\forall u \in \left\{1,...,U\right\}$.

\subsection{$\mathcal{M}_{\text{CC}} \ \& \ \mathcal{M}_{\text{CA}}$ Assignment}\label{section:MDDTHz:MCA}
Intuitively, the CAP-to-AP link has the priority over the CPU-to-CAP link to choose subcarriers. The reason can be explained as follows: (i) The CPU having more antennas and larger transmit power is capable of mitigating the interference between CAPs and overcoming the high path loss; (ii) On the contrary, CAP has to serve multiple APs in the cluster with relatively less antennas and transmit power, while the quality of CAP-to-AP link is also subject to the inter-cluster interference. Therefore, during the initialization process, it is reasonable to allocate the subcarriers with better channel quality to CAP-to-AP link, thereby guaranteeing the reliable fronthaul rate.    

In this regard, we adopt the unfair greedy algorithm to first allocate the subcarriers to CAP-to-AP link, and the remaining subcarriers are then assigned to CPU-to-CAP link.  Note that the maximum ratio transmission (MRT) precoding strategy is applied to reduce the computational complexity during the subcarrier selection. Next, we mathematically describe the process of the proposed approach. Let $\mathcal{M}_{\text{CC}}=\mathcal{M}_{\text{CA}}=\emptyset$.
Given that CAP $l$ can obtain $\pmb{h}_{l,q}[\tilde{m}], \ \forall q \in \mathcal{C}_l,\tilde{m} \in \mathcal{M}_{\text{FH}}$, the channel gain between CAP $l$ and AP $q$ over the $\tilde{m}$-th subcarrier is derived as $w_{l,q}[\tilde{m}]=\left\|\pmb{h}_{l,q}[\tilde{m}]\right\|_2^2$. Then, we find the subcarrier $\bar{m}$ having the maximum channel gain in $\mathcal{C}_l$, i.e., $\bar{m}=\argmax_{\tilde{m}}\sum_{q \in \mathcal{C}_l}w_{l,q}[\tilde{m}]$, and update $\mathcal{M}_{\text{CA}}$ with $\mathcal{M}_{\text{CA}}=\mathcal{M}_{\text{CA}}\bigcup \left\{\bar{m}\right\}$. All the CAPs sequentially select their optimal subcarriers and incorporate them into $\mathcal{M}_{\text{CA}}$, until the size of $\mathcal{M}_{\text{CA}}$ is satisfied. 

\section{Proposed Solution for the End-to-End Rate Optimization}
With the methods proposed in Section \ref{section:MDDTHz:relax}, the results of AP clustering, device selection and subcarrier sets assignment can be obtained for the particular $L$ and sizes of $\mathcal{M}_{\text{CC}}$ and $\mathcal{M}_{\text{CA}}$. Then, the original problem in \eqref{eq:MDD-THz:opt} is divided into two independent sub-problems, that is, the optimization of two fronthaul links relying on the power allocation and subcarrier distribution at the CPU and CAPs for conveying different devices' data, and the optimization of access links related to the power allocation at APs, which can be expressed as follows: 
\begin{align}\label{eq:MDD-THz:SE_sub1}
\left(\mathcal{P}_1\right) \ \ \ &\underset{\left\{\pmb{v}\right\},\left\{p_{l,q,u}\right\}}{\text{maximize}} \ \min_{u\in\left\{1,...,U\right\}}\left\{C_u^{\text{AD}}\right\}, \ \text{s.t.} \ ~~(\ref{eq:MDD-THz:opt}\text{j}),
\end{align}
and
\begin{subequations}
\label{eq:MDD-THz:SE_sub2}
\begin{align}
\left(\mathcal{P}_2\right) \ \ \ &\underset{\left\{\pmb{f}\right\},\left\{\pmb{w}\right\},\left\{\gamma\right\},\left\{p_{l,u}[m]\right\},\left\{p_{l,q,u}[\bar{m}]\right\}}{\text{maximize}} \ \min_{u\in\left\{1,...,U\right\}}\left\{C_u^{\text{CC}},C_u^{\text{CA}}\right\} \\
\text{s.t.} \ &~~(\ref{eq:MDD-THz:opt}\text{f}), (\ref{eq:MDD-THz:opt}\text{g}), (\ref{eq:MDD-THz:opt}\text{h}),(\ref{eq:MDD-THz:opt}\text{i}).
\end{align}
\end{subequations}

\subsection{The Optimization of DL Transmissions}

During the optimization of resource allocation, we assume that the regularized zero-forcing (RZF) precoding is applied. Then, we have {\scriptsize{$\pmb{V}^{\text{RZF}}_{l}=\left[\pmb{v}^{\text{RZF}}_{l,1},...,\pmb{v}^{\text{RZF}}_{l,\left|\mathcal{U}_l\right|}\right]=\pmb{H}_l\left(\pmb{H}_l^H\pmb{H}_l+\epsilon \pmb{I}_{\left|\mathcal{U}_l\right|}\right)^{-1}$}}, where $\small \pmb{v}^{\text{RZF}}_{l,u} \in \mathbb{C}^{\left|\mathcal{C}_l\right|N^{\text{AP}} \times 1}=\left[\left(\pmb{v}_{l,1,u}^{\text{RZF}}\right)^T,...,\left(\pmb{v}_{l,\left|\mathcal{C}_l\right|,u}^{\text{RZF}}\right)^T\right]^T$, $\small \pmb{H}_l \in \mathbb{C}^{\left|\mathcal{C}_l\right|N^{\text{AP}} \times \left|\mathcal{U}_l\right|}=\left[\pmb{h}_{l,1},...,\pmb{h}_{l,\left|\mathcal{U}_l\right|}\right]$, $\epsilon >0$ is the regularization factor. Note that $\pmb{v}_{l,1,u}^{\text{RZF}}$ is the normalized precoder following $\small \pmb{v}_{l,q,u}^{\text{RZF}}\leftarrow{\pmb{v}_{l,q,u}^{\text{RZF}}}/{\left\|\pmb{v}^{\text{RZF}}_{l,u}\right\|_2}$. Correspondingly, the optimization $\left(\mathcal{P}_1\right)$ can be reformed as
\begin{subequations}
\small
\label{eq:MDD-THz:SE_resub1}
\begin{align}
&\underset{\left\{p_{l,q,u}\right\}}{\text{maximize}} \min_{u\in\left\{1,...,U\right\}} \nonumber \\
&{\scriptsize \left\{\text{SINR}_u= \frac{\left|\sum_{l \in \mathcal{G}_u}\pmb{h}_{l,u}^H\tilde{\pmb{P}}_{l,u}\pmb{v}_{l,u}^{\text{RZF}}\right|^2}{\sum_{l \in \mathcal{G}_u}\sum_{u^{\prime} \in \mathcal{U}_l,u^{\prime}\neq u}\left|\pmb{h}_{l,u}^H\tilde{\pmb{P}}_{l,u^{\prime}}\pmb{v}^{\text{RZF}}_{l,u^{\prime}}\right|^2+\sum_{l^{\prime} \notin \mathcal{G}_u} \sum_{u^{\prime} \in \mathcal{U}_{l^{\prime}}}\left|\pmb{h}_{l^{\prime},u^{\prime}}^H\tilde{\pmb{P}}_{l^{\prime},u^{\prime}}\pmb{v}^{\text{RZF}}_{l^{\prime},u^{\prime}}\right|^2+\sigma^2_u}\right\}} \\
&\text{s.t.} \ \sum_{u \in \mathcal{U}_l}p_{l,q,u}\left\|\pmb{v}^{\text{RZF}}_{l,q,u}\right\|_2^2 \leq P_{\text{AP}}, \forall l \in \left\{1,...,L\right\}, q \in \left\{1,...,Q\right\}.
\end{align}
\end{subequations}
Upon introducing a slack variable $\chi$ associated with SINR, an alternative optimization problem can be formulated as
\begin{subequations}
\label{eq:MDD-THz:SE_resub11}
\small
\begin{align}
&\underset{\left\{p_{l,q,u}\right\}}{\text{minimize}} \ \sum_{l=1}^L \sum_{u \in \mathcal{U}_l}\left\|\pmb{P}_{l,u}\right\|_{\text{F}}, \\
&\text{s.t.} \ {\scriptsize \left\|\left[\underbrace{...,\pmb{h}_{l,u}^H\tilde{\pmb{P}}_{l,u^{\prime}}\pmb{v}^{\text{RZF}}_{l,u^{\prime}},...,}_{\sum_{l \in \mathcal{G}_u}(\left|\mathcal{U}_l\right|-1) \ \text{elements}}\underbrace{...,\pmb{h}_{l^{\prime},u^{\prime}}^H\tilde{\pmb{P}}_{l^{\prime},u^{\prime}}\pmb{v}^{\text{RZF}}_{l^{\prime},u^{\prime}},...,}_{\sum_{l^{\prime} \notin \mathcal{G}_u}\left|\mathcal{U}_{l^{\prime}}\right| \ \text{elements}}\sqrt{\sigma^2_u}\right]\right\| \leq \left|\sum_{l \in \mathcal{G}_u}\pmb{h}_{l,u}^H\tilde{\pmb{P}}_{l,u}\pmb{v}_{l,u}^{\text{RZF}}\right|\sqrt{\frac{1}{\chi}}},\ \forall u \in \left\{1,...,U\right\},  \nonumber \\ 
&\text{and} \ (\ref{eq:MDD-THz:SE_resub1}\text{b}),
\end{align}
\end{subequations}
where the introduction of (\ref{eq:MDD-THz:SE_resub11}\text{a}) aims to accelerate the convergence rate. Accordingly, the problem \eqref{eq:MDD-THz:SE_resub1} can be equivalently solved by finding the maximum value within the set $\left\{\chi\right\}$, which makes the second-order cone optimization in \eqref{eq:MDD-THz:SE_resub11} feasible. In order to search the optimal $\chi$, the upper bound of $\chi$ is required to be determined first, which is described as\vspace{-0.3cm}
\begin{equation}
\small
\chi_{\text{upper}}=\underset{{u \in \left\{1,...,U\right\}}}{\min}\left\{\underset{\left\{p_{l,q,u}\right\}}{\text{maximize}} {\left|\sum_{l \in \mathcal{G}_u}\pmb{h}_{l,u}^H\tilde{\pmb{P}}_{l,u}\pmb{v}_{l,u}^{\text{RZF}}\right|^2 }/{\sigma^2_u}, \ \text{s.t.} \ (\ref{eq:MDD-THz:SE_resub1}\text{b}) \right\}.
\end{equation}
Then, the bisection method (e.g., \cite[Algorithm 7.5]{demir2021foundations}) can be applied to derive the solution. 

\subsection{The Optimization of Fronthaul Transmissions}\label{section:MDDTHz:front}
The objective of $\left(\mathcal{P}_2\right)$ is to maximize the fronthaul rates under the fairness constraint. As seen in \eqref{eq:MDD-THz:SE_sub2}, given that the assignment of  the subcarrier sets for two MDD-based fronthaul links is fixed, the optimization of $C^{\text{CC}}$ and $C^{\text{CA}}$ can be implemented in a parallel way. The RZF precoding is assumed to be adopted during resource allocation over two fronthaul links. More specifically, we have $\small \pmb{F}^{\text{RZF}}[m] \in \mathbb{C}^{N^{\text{CPU}} \times L}=\left[\pmb{f}^{\text{RZF}}_{1,u}[m],...,\pmb{f}^{\text{RZF}}_{L,u}[m]\right]=\pmb{H}[m]\left(\pmb{H}^H[m]\pmb{H}[m]+\epsilon \pmb{I}_L\right)^{-1}$ and $\small \pmb{W}^{\text{RZF}}_l[\bar{m}] \in \mathbb{C}^{N^{\text{AP}} \times \left|\mathcal{C}_l\right|}=\left[\pmb{w}^{\text{RZF}}_{l,1,u}[\bar{m}],...,\pmb{w}^{\text{RZF}}_{l,\left|\mathcal{C}_l\right|,u}[\bar{m}]\right]=$\\$\small \pmb{H}_l[\bar{m}]\left(\pmb{H}_l^H[\bar{m}]\pmb{H}_l[\bar{m}]+\epsilon \pmb{I}_{\left|\mathcal{C}_l\right|}\right)^{-1}$, where $\small \pmb{H}[m] \in \mathbb{C}^{N^{\text{CPU}} \times L}=\left[\pmb{h}_{1}[m],...,\pmb{h}_{L}[m]\right]$, $\small \pmb{H}_l[\bar{m}] \in \mathbb{C}^{N^{\text{AP}} \times \left|\mathcal{C}_l\right|}=\left[\pmb{h}_{l,1}[\bar{m}],...,\pmb{h}_{l,\left|\mathcal{C}_l\right|}[\bar{m}]\right]$. Note that each column of $\pmb{F}^{\text{RZF}}[m]$ and $\pmb{W}^{\text{RZF}}_l[\bar{m}]$ 
has been normalized, i.e., $\small \pmb{f}^{\text{RZF}}_{l,u}[m]\leftarrow \pmb{f}^{\text{RZF}}_{l,u}[m]/\left\|\pmb{f}^{\text{RZF}}_{l,u}[m]\right\|_2$ and $\small \pmb{w}^{\text{RZF}}_{l,q,u}[\bar{m}]\leftarrow \pmb{w}^{\text{RZF}}_{l,q,u}[\bar{m}]/\left\|\pmb{w}^{\text{RZF}}_{l,q,u}[\bar{m}]\right\|_2$. Then, the optimization of CPU-to-CAP link can be formulated as
\begin{subequations}
\label{eq:MDD-THz:SE_sub3}
\begin{align}
&\left(\mathcal{P}_{3}\right) \ \underset{\left\{p_{l,u}[m]\right\},\left\{\gamma_{l,u,m}\right\}}{\text{maximize}} \min_{u\in\left\{1,...,U\right\}, l\in \mathcal{G}_u}\left\{C_{l,u}^{\text{CC}}\right\} \\
&\text{s.t.} \ \sum_{m \in \mathcal{M}_{\text{CC}}}\sum_{l=1}^L \sum_{u=1}^U \gamma_{l,u,m}p_{l,u}[m] \leq P_{\text{CPU}}, \ \text{and}
~~(\ref{eq:MDD-THz:opt}\text{f}).
\end{align}
\end{subequations}
Obviously, $\left(\mathcal{P}_{3}\right)$ is a non-convex optimization problem due to the nature of objective function and binary variables $\left\{\gamma_{l,u,m}\right\}$ along with constraint (\ref{eq:MDD-THz:opt}\text{f}). To this end, in what follows, we present a tractable form of (\ref{eq:MDD-THz:SE_sub3}) by exploiting the relationship between continuous and binary variables, based on which the combination of quadratic transform (QT) and bisection method can be leveraged to solve it in an iterative way.

\subsubsection{Binary Reduction of $\gamma$}
The binary variable $\gamma$ and the continuous variable $p$ are strongly coupled, and their relationship can be demonstrated by the following Lemma:
\begin{lemma}\label{lemma:MDDTHz:mup}
There exist only two potential combinations of $\gamma_{l,u,m}$ and $p_{l,u}[m]$ in the optimal solution to $\left(\mathcal{P}_{3}\right)$, which are $(\gamma_{l,u,m}^\ast,p_{l,u}^\ast[m])\in \left\{(0,p_{l,u}[m]= 0),(1,p_{l,u}[m]\neq 0)\right\}$, where $p_{l,u}[m]\neq 0$.
\end{lemma} 

\begin{proof}
Refer to \cite[Lemma 1]{li2022spectral}
\end{proof} 
Based on Lemma 1, only $p_{l,u}[m]$ is needed during the optimization process. The feasible $\gamma_{l,u,m}^\ast$ can be subsequently obtained from the optimal $p_{l,u}^\ast[m]$, which is given as 
\begin{equation}\label{eq:MDDTHz:ldmpower}
\gamma_{l,u,m}^\ast=
\begin{cases}
0, &\frac{p_{l,u}^\ast[m]}{P_{\text{CPU}}}<\zeta \\
1, &\frac{p_{l,u}^\ast[m]}{P_{\text{CPU}}}\geq \zeta
\end{cases},
\end{equation} 
where $\zeta$ is a very small number, implying that a small value of $p_{l,u}^\ast[m]$ can be deemed as zero. 

\subsubsection{Tractable Constraint}
According to Lemma 1, (\ref{eq:MDD-THz:opt}\text{f}) can be substituted with the new constraint only related to $p_{l,u}[m]$, which is given by
\begin{equation}\label{eq:MDDTHz:newcon}
\left\|\pmb{p}_l[m]\right\|_0 \leq 1, \forall l \in \left\{1,...,L\right\}, m \in \mathcal{M}_{\text{CC}},
\end{equation}   
where $\pmb{p}_l[m]=\left[p_{l,1}[m],...,p_{l,\left|\mathcal{U}_l\right|}[m]\right]^T$. As $L^0$-norm is non-convex, it has to be approximated by the tractable convex functions. In particular, we resort to the exp-based surrogate function for the sake of better performance, which can be expressed as \cite{song2015sparse,sun2021joint}\vspace{-0.5cm}
\begin{equation}\vspace{-0.3cm}
\left\|\pmb{p}_l[m]\right\|_0  \approx \sum_{u=1}^{\left|\mathcal{U}_l\right|}\left(1-\exp\left(-\psi p_{l,u}[m]\right)\right) := \Xi(\pmb{p}_l[m]),
\end{equation}
where $\psi \gg 0$ is a sufficiently large value to guarantee the approximation. Then, since $\Xi(\pmb{p}_l[m])$ is concave with respect to $p_{l,u}[m]\in \mathbb{R}^+$, based on its subgradient inequality, we have
\begin{align}
&\Xi(\pmb{p}_l[m]) \leq \Xi(\pmb{p}_l^{(i)}[m])+\sum_{u=1}^{\left|\mathcal{U}_l\right|}\left[\psi\exp\left(-\psi p_{l,u}^{(i)}[m]\right)\left(p_{l,u}[m]-p_{l,u}^{(i)}[m]\right)\right] :=\tilde{\Xi}^{(i)}(\pmb{p}_l[m]),
\end{align}
where $p_{l,u}^{(i)}[m]$ is the feasible value of $p_{l,u}[m]$ at the $i$-th iteration. Consequently, $\tilde{\Xi}(\pmb{p}_l[m])$ is the convex function with respect to $p_{l,u}[m]$, and the constraint 
\begin{align}\label{eq:MDDTHz:finacon}
\tilde{\Xi}^{(i)}(\pmb{p}_l[m])\leq 1, \forall l \in \left\{1,...,L\right\}, m \in \mathcal{M}_{\text{CC}}
\end{align}
can substitute for \eqref{eq:MDDTHz:newcon}.

\subsubsection{Tractable Objective Function}   
{\small $C_{l,u}^{\text{CC}}\!=\!b^{\text{CC}}\!\sum_{m \in \mathcal{M}_{\text{CC}}}\!\log\left(\!1\!+\!{A_{l,u,m}\left(\pmb{p}[m]\right)}/{B_{l,u,m}\left(\pmb{p}[m]\right)}\!\right)$} in the objective function (\ref{eq:MDD-THz:SE_sub3}a) belongs to the multiple-ratio fractional programming (MRFP) problem, where {\small $A_{l,u,m}\left(\pmb{p}[m]\right)=p_{l,u}[m]\left|\pmb{h}_l^H[m]\pmb{f}^{\text{RZF}}_{l,u}[m]\right|^2$, $B_{l,u,m}\left(\pmb{p}[m]\right)\!=\!\sum_{ l^{\prime}\neq l }\!\sum_{u^{\prime} \in \mathcal{U}_{l^{\prime}}}\! \times$\\ $p_{l^{\prime},u^{\prime}}[m]\!\left|\pmb{h}_l^H[m]\pmb{f}^{\text{RZF}}_{l^{\prime},u^{\prime}}[m]\right|^2\!+\!\left|\mathbb{E}\left\{\text{SI}_l\right\}\right|^2\!+\!\sigma_{\text{n}}^2$}, and {\small $\pmb{p}[m]\in \mathbb{R}^{\sum_{l}\left|\mathcal{U}_l\right| \times 1}$} denotes the power allocation vector for the $m$-th subcarrier. Since $\log(\cdot)$ function is non-decreasing and $\frac{A_{l,u,m}\left(\pmb{p}[m]\right)}{B_{l,u,m}\left(\pmb{p}[m]\right)}$ satisfies the convex-over-concave form, $\left(\mathcal{P}_{3}\right)$ can be equivalently formulated as
\begin{subequations}
\label{eq:MDD-THz:SE_sub3_1}
\begin{align}
&\left(\mathcal{P}_{3-1}\right) \ \underset{\left\{p_{l,u}[m]\right\},\left\{z_{l,u,m}^{(i)}\right\}}{\text{maximize}} \tilde{\chi} \\
&\text{s.t.} \ z_{l,u,m}^{(i)} \in \mathbb{R}, \ p_{l,u}^{(i)}[m] \in \mathbb{R}^+, \ \forall l \in \left\{1,...,L\right\}, u \in \left\{1,...,U\right\}, m \in \mathcal{M}_{\text{CC}}, \\
&b^{\text{CC}}\sum_{m \in \mathcal{M}_{\text{CC}}}\log\left(1+2z_{l,u,m}^{(i)}\sqrt{A_{l,u,m}\left(\pmb{p}[m]\right)}-\left(z_{l,u,m}^{(i)}\right)^2B_{l,u,m}\left(\pmb{p}[m]\right)\right)\geq \tilde{\chi}, \nonumber \\ 
&\forall u\in\left\{1,...,U\right\}, l\in \mathcal{G}_u, \\
&\sum_{m \in \mathcal{M}_{\text{CC}}}\sum_{l=1}^L \sum_{u=1}^U p_{l,u}[m] \leq P_{\text{CPU}}, \ \text{and} ~~(\ref{eq:MDDTHz:finacon}),
\end{align}
\end{subequations}
where $z_{l,u,m}^{(i)}$ is the feasible value at the $i$-th iteration. 

To verify this, we can first rewrite $\left(\mathcal{P}_{3}\right)$ as the maximization of $\tilde{\chi}$ subject to (\ref{eq:MDD-THz:SE_sub3_1}d) and $\tilde{\chi}\leq b^{\text{CC}}\sum_{m \in \mathcal{M}_{\text{CC}}}\log\left(1+{A_{l,u,m}\left(\pmb{p}[m]\right)}/{B_{l,u,m}\left(\pmb{p}[m]\right)}\right)$. According to the principle of QT \cite{shen2018fractional}, the latter constraint can be reformed as 
\begin{equation}
\small 
\tilde{\chi}\leq {\text{maximize}}_{\left\{z_{l,u,m}^{(i)}\right\}} \ \ b^{\text{CC}}\sum_{m \in \mathcal{M}_{\text{CC}}}\log\left(1+2z_{l,u,m}^{(i)}\sqrt{A_{l,u,m}\left(\pmb{p}[m]\right)}-\left(z_{l,u,m}^{(i)}\right)^2B_{l,u,m}\left(\pmb{p}[m]\right)\right).
\end{equation}
As this new constraint is a less-than-max inequality, ${\text{maximize}}_{\left\{z_{l,u,m}^{(i)}\right\}}$ can be integrated into ${\text{maximize}}_{\left\{p_{l,u}[m]\right\}}$, as in (\ref{eq:MDD-THz:SE_sub3_1}a). 

\begin{algorithm}
\caption{The Optimization of CPU-to-CAP fronthaul} 
\label{MDDTHz:al2}
\scriptsize
\KwIn{$\mathcal{M}_{\text{CC}}, \left\{\mathcal{C}_l\right\}_{l=1}^L, \left\{\mathcal{U}_l\right\}_{l=1}^L, \left\{\pmb{F}^{\text{RZF}}[m]\right\}_{m=1}^{\left|\mathcal{M}_{\text{CC}}\right|}, \left\{\pmb{H}[m]\right\}_{m=1}^{\left|\mathcal{M}_{\text{CC}}\right|}$;}
\textbf{Initialization:} Randomly set $\left\{\gamma_{l,u,m}^{(0)}\right\}$ under the constraint (\ref{eq:MDD-THz:SE_sub3}c)\; 
Initialize $p_{l,u}^{(0)}[m]=\frac{P_{\text{CPU}}}{\left|\mathcal{M}_{\text{CC}}\right|L}, \forall l \in \left\{1,...,L\right\}, u \in \left\{1,...,U\right\}, m \in \mathcal{M}_{\text{CC}}$\;
Compute $z_{l,u,m}^{(0)}=\frac{\sqrt{A_{l,u,m}\left(\pmb{p}^{(0)}[m]\right)}}{B_{l,u,m}\left(\pmb{p}^{(0)}[m]\right)}$ and set $i=0$\;
Set $\kappa$, $\tilde{\chi}_{\text{lower}}=0$ and $\tilde{\chi}_{\text{upper}}=\underset{u\in\left\{1,...,U\right\}, l\in \mathcal{G}_u}{\min}\underset{\left\{p_{l,u}[m]\right\}}{\text{maximize}}\left\{b^{\text{CC}}\sum_{m\in \left|\mathcal{M}_{\text{CC}}\right|}\log\left(1+\frac{p_{l,u}[m]\left|\pmb{h}_l^H[m]\pmb{f}^{\text{RZF}}_{l,u}[m]\right|^2}{\sigma^2_n}\right), \text{s.t. }\sum_{m\in \mathcal{M}_{\text{CC}}}p_{l,u}[m]\leq P_{\text{CPU}}\right\}$\;
\QB{}{
\Repeat{\em\text{convergence}}{
Set $\tilde{\chi}=\frac{\tilde{\chi}_{\text{upper}}+\tilde{\chi}_{\text{lower}}}{2}$\;
For the fixed $\left\{z_{l,u,m}^{(i)}\right\}$, update $\left\{p_{l,u}^{(i+1)}[m]\right\}$ by solving the following optimization problem: \\ 
\vspace{-0.3cm}
\begin{flalign}\label{eq:MDD-THz:QB}
&\underset{\left\{p_{l,u}[m]\right\}}{\text{minimize}}\sum_{m \in \mathcal{M}_{\text{CC}}}\left\|\pmb{p}[m]\right\|_1, \text{s.t.}  \ (\ref{eq:MDD-THz:SE_sub3_1}c) \ \text{and} \ (\ref{eq:MDD-THz:SE_sub3_1}d)\vspace{-0.3cm}
\end{flalign}
\eIf{\em{(\ref{eq:MDD-THz:QB}) \text{is feasible}}}{Update $\tilde{\chi}_{\textup{lower}}=\tilde{\chi}$ and $\left\{p_{l,u}^{\ast}[m]\right\}=\left\{p_{l,u}^{(i+1)}[m]\right\}$\;
Compute $\left\{\gamma_{l,u,m}^{(i+1)}\right\}$ as in \eqref{eq:MDDTHz:ldmpower} and update $\left\{\gamma_{l,u,m}^{\ast}\right\}=\left\{\gamma_{l,u,m}^{(i+1)}\right\}$\;
Update $z_{l,u,m}^{(i+1)}=\frac{\sqrt{A_{l,u,m}\left(\pmb{p}^{(i+1)}[m]\right)}}{B_{l,u,m}\left(\pmb{p}^{(i+1)}[m]\right)}$ and set $i=i+1$\;
\If{$\tilde{\chi}_{\textup{upper}}-\tilde{\chi}_{\textup{lower}}<\kappa
$}{$\tilde{\chi}_{\text{upper}}=\tilde{\chi}_{\text{upper}}+\chi_{\text{com}}$;}
}{$\tilde{\chi}_{\textup{upper}}=\tilde{\chi}$\;
} 
}
}
\KwOut{$C^{\text{CC}}= \min_{u\in\left\{1,...,U\right\}, l\in \mathcal{G}_u}\left\{b^{\text{CC}}\sum_{m \in \mathcal{M}_{\text{CC}}}\log\left(1+\frac{A_{l,u,m}\left(\pmb{p}[m]\right)}{B_{l,u,m}\left(\pmb{p}[m]\right)}\right)| \left\{p_{l,u}^{\ast}[m]\right\} \right\}$ and $\left\{\gamma_{l,u,m}^{\ast}\right\}$}
\end{algorithm}

Accordingly, $\left(\mathcal{P}_{3-1}\right)$ can be solved by applying the QT and bisection methods to iteratively optimize the primal variable $\left\{p_{l,u}[m]\right\}$ and the auxiliary variables $z_{l,u,m}^{(i)}$ and $p_{l,u}^{(i)}[m]$. To be more specific, we can initialize the set $\left\{p_{l,u}^{(0)}[m]\right\}$ by equally allocating the power and subcarriers among devices, and then the optimal set $\left\{z_{l,u,m}^{(0)}\right\}$ can be obtained based on the QT principle, i.e., $z_{l,u,m}^{(0)}=\frac{\sqrt{A_{l,u,m}\left(\pmb{p}^{(0)}[m]\right)}}{B_{l,u,m}\left(\pmb{p}^{(0)}[m]\right)}$. Next, we solve the optimization problem, i.e., ${\text{minimize}}_{\left\{p_{l,u}[m]\right\}}\sum_{m \in \mathcal{M}_{\text{CC}}}\left\|\pmb{p}[m]\right\|_1$, subject to (\ref{eq:MDD-THz:SE_sub3_1}c)-(\ref{eq:MDD-THz:SE_sub3_1}e), with $\tilde{\chi}=\frac{\tilde{\chi}_{\text{upper}}+\tilde{\chi}_{\text{lower}}}{2}$, to obtain the optimal set $\left\{p_{l,u}^{(1)}[m]\right\}$. It is noteworthy that after the $i$-th iteration with the $\tilde{\chi}$ limited to $\tilde{\chi}^{(i)}_{\text{upper}}$ and new $\left\{p_{l,u}^{(i)}[m]\right\}$, the updated $\left\{z_{l,u,m}^{(i)}\right\}$ may lead to that the maximum of the left-hand side of (\ref{eq:MDD-THz:SE_sub3_1}c) is larger than $\tilde{\chi}^{(i)}_{\text{upper}}$, due to the nondecreasing feature of QT \cite{shen2018fractional}. Hence, we introduce a compensation factor $\chi_{\text{com}}$ to overcome this problem, that is, after each iteration, if $\tilde{\chi}_{\textup{upper}}-\tilde{\chi}_{\textup{lower}}<\kappa$, we set $\tilde{\chi}^{(i)}_{\text{upper}}=\tilde{\chi}^{(i)}_{\text{upper}}+\chi_{\text{com}}$, where $\kappa$ is a predefined very small value. The above-mentioned processes can be iteratively implemented until convergence is achieved. In summary, the overall optimization is stated as Algorithm \ref{MDDTHz:al2}. 

Analogously, the optimization of CAP-to-AP fronthaul can be solved using the same method proposed in Algorithm \ref{MDDTHz:al2}, since the optimization of these two fronthaul links have the similar formulation and constraints as shown in \eqref{eq:MDD-THz:CCSINR}, \eqref{eq:MDD-THz:CASINR} and \eqref{eq:MDD-THz:opt}.

\subsection{Overall End-to-End Optimization}
The above-proposed algorithms for the optimization of access and fronthaul links are based on the given results of AP clustering, device selection and the assignment of subcarrier sets. Next, we aim to iteratively assign different numbers of subcarriers between two MDD-based fronthaul links, and select the optimal number of AP clusters in order to achieve the end-to-end optimization in \eqref{eq:MDD-THz:opt}.

\begin{algorithm}
\caption{Overall End-to-End Rate Optimization} 
\label{MDDTHz:al3}
\scriptsize
\textbf{Initialization:} Set $\left|\mathcal{M}_{\text{CA}}\right|=\left|\mathcal{M}_{\text{CC}}\right|=\frac{1}{2}\left|\mathcal{M}_{\text{FH}}\right|$, $L_{\text{ini}}=\frac{U}{U_{\text{max}}}$, $i=1$, $l_{\text{step}}$ and $m_{\text{step}}$\; 
\While{$L_{\text{ini}}+(i-1)l_{\textup{step}}\leq Q$}{
Obtain $\left\{\mathcal{C}_l\right\}_{l=1}^L, \left\{\mathcal{U}_l\right\}_{l=1}^L$ and $\left\{\mathcal{G}_u\right\}_{u=1}^U$ using the heuristic methods in Section \ref{section:MDDTHz:APC}\;
Compute $C_i^{\text{AD}}$ by solving \eqref{eq:MDD-THz:SE_resub11}\;
\Repeat{$\left|C^{\text{CC}}_i-C^{\text{CA}}_i\right|\leq\kappa^{'}$}{
Set $j=1$\;
Obtain $\mathcal{M}_{\text{CC}}$ and $\mathcal{M}_{\text{CA}}$ using the heuristic method in Section \ref{section:MDDTHz:MCA}\;
Compute $C_i^{\text{CC}}$ and $C_i^{\text{CA}}$ by Algorithm \ref{MDDTHz:al2}\;
\eIf{$C_i^{\text{CC}}>C_i^{\text{CA}}$}{$\left|\mathcal{M}_{\text{CA}}\right|=\left|\mathcal{M}_{\text{CA}}\right|+\vartheta^{(j-1)} \cdot m_{\text{step}}$;}{$\left|\mathcal{M}_{\text{CA}}\right|=\left|\mathcal{M}_{\text{CA}}\right|-\vartheta^{(j-1)} \cdot m_{\text{step}}$;}
Update $j=j+1$\;
}
Obtain $C_i^{\text{FH}}=\min\left\{C_i^{\text{CC}},C_i^{\text{CA}}\right\}$\;
Update $i = i+1$\;
}
\KwOut{$C^\ast = \underset{i}{\text{maximize}} \min\left\{C_i^{\text{AD}},C_i^{\text{FH}}\right\}$}
\end{algorithm}

\subsubsection{Iterative Assignment of Subcarrier Sets}
In Section \ref{section:MDDTHz:front}, given the fixed assignment of subcarrier sets, the optimization of each fronthaul can be solved by Algorithm \ref{MDDTHz:al2}. However, considering that the total fronthaul rates is calculated as $C^{\text{FH}}=\min\left\{C^{\text{CC}},C^{\text{CA}}\right\}$, 
it is intuitive to highlight the following trade-off: Provided that the transmit power is fixed, the more subcarriers are allocated, the higher rates can be achieved, owing to the subcarrier diversity gain. Therefore, in order to maximize $C^{\text{FH}}$, the sizes of two subcarriers sets (i.e., $\mathcal{M}_{\text{CC}}$ and $\mathcal{M}_{\text{CA}}$) can be iteratively increased or decreased by following a pre-defined step namely $m_{\text{step}}$, according to the relationship between two fronthaul rates. To be more specific, after optimizing the two fronthaul links separately for given $\left|\mathcal{M}_{\text{CA}}\right|=\frac{1}{2}\left|\mathcal{M}_{\text{FH}}\right|$, the iterative assignment of subcarrier sets commences. If $C^{\text{CC}}>C^{\text{CA}}$, the size of $\mathcal{M}_{\text{CA}}$ is increased to $\left|\mathcal{M}_{\text{CA}}\right|=\left|\mathcal{M}_{\text{CA}}\right|+\vartheta^{(j-1)} \cdot m_{\text{step}}$, where $\vartheta$ is a decay factor, $j$ is the current iteration index. Then, the subcarrier assignment method proposed in Section \ref{section:MDDTHz:MCA} is re-calculated. Otherwise, if $C^{\text{CC}}<C^{\text{CA}}$, the subcarrier assignment process is re-run with the size of $\mathcal{M}_{\text{CA}}$ decreased to $\left|\mathcal{M}_{\text{CA}}\right|=\left|\mathcal{M}_{\text{CA}}\right|-\vartheta^{(j-1)} \cdot m_{\text{step}}$. This process ends up with $\left|C^{\text{CC}}-C^{\text{CA}}\right|\leq\kappa^{'}$.         

\subsubsection{Iterative Optimization of Cluster Number}
Given that the AP clustering is accomplished by the heuristic method proposed in Section \ref{section:MDDTHz:APC}, the trade-off between access and two fronthaul links only depends on the number of AP clusters $L$, which can be highlighted as follows: (i) Changing the number of AP clusters may impose distinct impact on two fronthaul links. In particular, a smaller $L$ results in more power received at each CAP and hence the CPU-to-CAP fronthaul rates are increased. By contrast, a smaller $L$ means that the average number of APs in each AP cluster is increased, which leads to less power received at APs and causes lower CAP-to-AP fronthaul rates. (ii) On the other hand, the number of AP clusters has an ambiguous impact on the AP-to-Device access links. To be more specific, when $L$ is smaller, the total number of APs serving a particular device $u$ is increased. However, it is possible that most of the APs have negligibly small channel gains to the device $u$, and hence the power transmitted to the device $u$ may cause extra interference to the nearby devices. By contrast, when $L$ becomes larger, the number of APs within each cluster is decreased. Then, with the constraints of $C_{\text{max}}$ and $U_{\text{max}}$, each device may only be served by its nearby AP clusters, and each AP cluster only allocates power to the devices with good channel conditions. In this case, the achievable rate of access links relies on the specific distributions of APs and devices, rather than the number of AP clusters.

According to the above analysis, the overall end-to-end optimization can be implemented as follows: (i) We firstly determine the range of $L$ to be $\frac{U}{U_{\text{max}}}\leq L \leq Q$, where $\frac{U}{U_{\text{max}}}\leq L $ guarantees that every device is served by at least one AP cluster\footnote{Note that, this value range is only applicable to DC, as it is capable of dynamically adjusting the AP clustering. By contrast, the implementation of AP grouping in SC and IDSC rely on the sub-area division, and for the sake of practical application, only few possible numbers of AP clusters will be considered during simulations in Section \ref{sec:MDDTHz:sim}. }. Then, $L$ is initialized to $L_{\text{ini}}=\frac{U}{U_{\text{max}}}$ and $\left|\mathcal{M}_{\text{CA}}\right|=\frac{1}{2}\left|\mathcal{M}_{\text{FH}}\right|$; (ii) Next, we set $i=1$, and $L_i=L_{\text{ini}}+(i-1)l_{\text{step}}$, where $l_{\text{step}}$ is a pre-defined step. Then, $C_i^{\text{CC}}, C_i^{\text{CA}}$ and $C_i^{\text{AD}}$ can be obtained, followed by the proposed iterative assignment of subcarrier sets. The optimal fronthaul rate is derived as $C_i^{\text{FH}}=\argmax_{\left\{\mathcal{M}_{\text{CC}},\mathcal{M}_{\text{CA}}\right\}}\min\left\{C_i^{\text{CC}}, C_i^{\text{CA}}\right\}$; (iii) The process continues until $L_i=Q$, and the optimal end-to-end rate can be obtained as 
$C^\ast = \underset{i}{\text{maximize}} \min\left\{C_i^{\text{AD}},C_i^{\text{FH}}\right\}$.
Algorithm \ref{MDDTHz:al3} summarizes the overall end-to-end optimization process.


\section{The Advanced Design of MDD Frame Structure}
In this section, followed by the proposed solution for jointly optimizing two fronthaul and access links, we integrate it into the specifically designed MDD frame structure so as to further improve the end-to-end rate. In order to exhibit the advantage of MDD, the TDD-based frame structure compatible with our designed two-tier wireless fronthaul architecture is also presented as a benchmark.

\begin{figure}[tbp]\vspace{-0.5cm}
\centering
\begin{minipage}[t]{0.45\textwidth}
\centering
\includegraphics[width=0.95\linewidth]{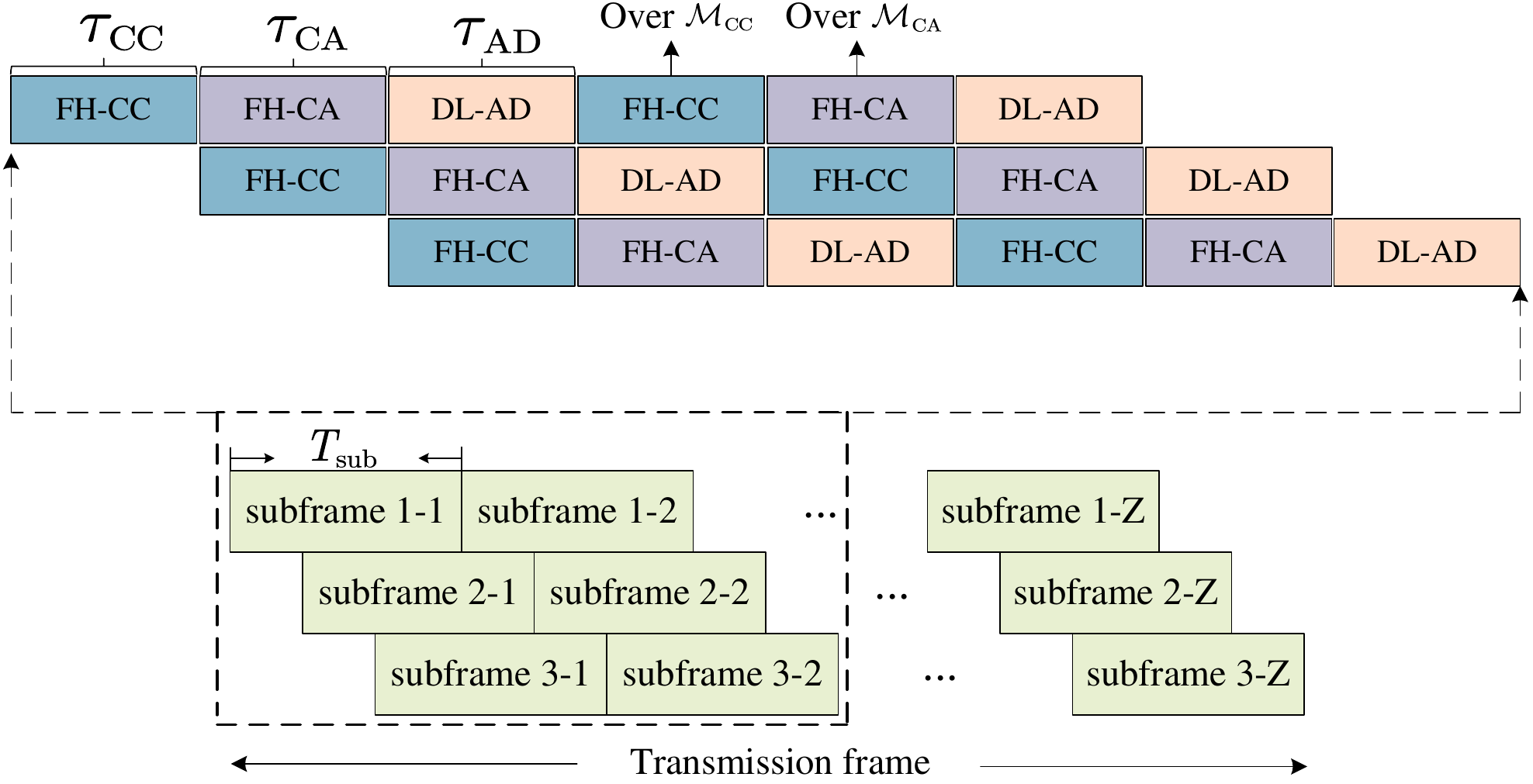}
\vspace{-0.5cm}
\caption{Frame structure of MDD-enbaled THz fronthaul.}
\label{figure-MDDTHz-MDDF}
\end{minipage}
\begin{minipage}[t]{0.45\textwidth}
\centering
\includegraphics[width=0.95\linewidth]{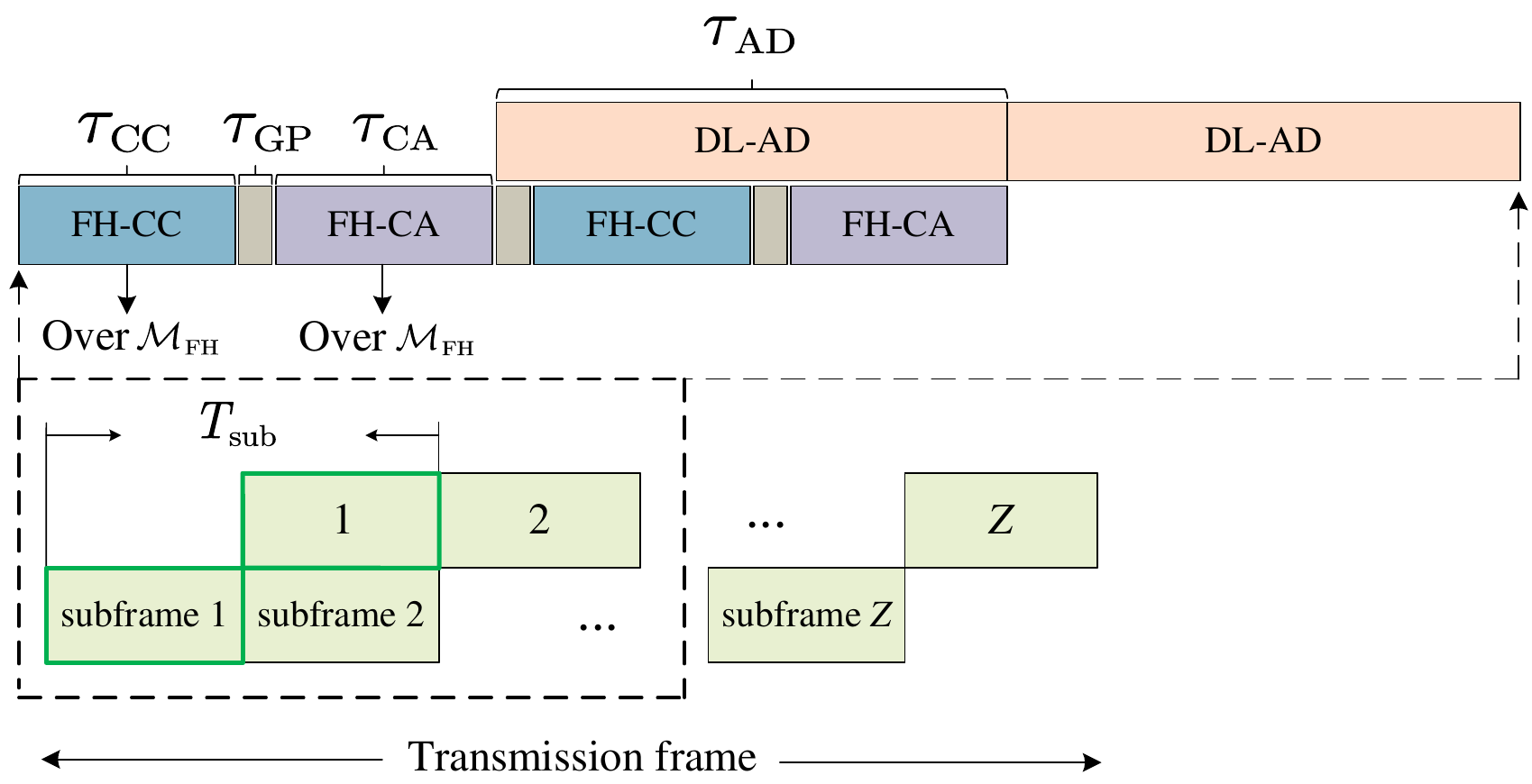}
\vspace{-0.5cm}
\caption{Frame structure of TDD-enbaled THz fronthaul.}
\label{figure-MDDTHz-TDDF}
\end{minipage}\vspace{-0.5cm}
\end{figure}

\subsection{Frame Structure for MDD-Enabled Fronthaul System}
Since two fronthaul transmissions are implemented in the MDD mode over the THz band and DL communications take place over the sub-6 GHz band, the MDD frame structure can be designed in a parallel way, as shown in Fig. \ref{figure-MDDTHz-MDDF}. Specifically, we consider a frame of transmission period within the coherence time, which is split into three parallel streams and each of them includes $Z$ subframes. Let the duration of one subframe be $T_{\text{sub}}$. The CPU-to-CAP, CAP-to-AP and AP-to-Device transmission durations are assumed to account for the same proportion within one subframe, i.e., $\tau_{\text{CC}}=\tau_{\text{CA}}=\tau_{\text{AD}}=\frac{1}{3}$, and the total duration of transmission frame is equal to $(Z+{2}/{3})T_{\text{sub}}$. Then, when $Z\rightarrow \infty$, the rates of two fronthaul and access links within one transmission frame are $\frac{Z C^{\text{CC}}}{(Z+2/{3})}\approx C^{\text{CC}},\frac{Z C^{\text{CA}}}{(Z+{2}/{3})}\approx C^{\text{CA}}$ and $\frac{Z C^{\text{AD}}}{(Z+{2}/{3})}\approx C^{\text{AD}}$, respectively. Therefore, the end-to-end rate optimization within one transmission frame is equal to the problem in \eqref{eq:MDD-THz:opt}, and can be solved by Algorithm \ref{MDDTHz:al3}.

Notably, in order to avoid the transmission collision, the assumption of $\tau_{\text{CC}}=\tau_{\text{CA}}=\tau_{\text{AD}}$ is essential to our proposed MDD frame structure. In this case, although the rate optimization between fronthaul and access links can not be further tuned via optimizing the proportion of their durations after solving \eqref{eq:MDD-THz:opt}, the proposed frame structure can significantly improve the end-to-end rates with the aid of three simultaneous data streams. 

\begin{algorithm}
\caption{The Overall End-to-End Rate Optimization for TDD-Enabled Fronthaul System} 
\label{MDDTHz:al4}
\scriptsize
\textbf{Initialization:} Set $\left|\mathcal{M}_{\text{CA}}\right|=\left|\mathcal{M}_{\text{CC}}\right|=\left|\mathcal{M}_{\text{FH}}\right|$, $L_{\text{ini}}=\frac{U}{U_{\text{max}}}$, $i=1$ and $l_{\text{step}}$\; 
\While{$L_{\text{ini}}+(i-1)l_{\textup{step}}\leq Q$}{
Obtain $\left\{\mathcal{C}_l\right\}_{l=1}^L, \left\{\mathcal{U}_l\right\}_{l=1}^L$ and $\left\{\mathcal{G}_u\right\}_{u=1}^U$ using the heuristic methods in Section \ref{section:MDDTHz:APC}\;
Compute $C_i^{\text{AD}}$ by solving \eqref{eq:MDD-THz:SE_resub11}\;
Compute $C_i^{\text{CC}}$ and $C_i^{\text{CA}}$ by Algorithm \ref{MDDTHz:al2}\;
Solve the convex problem (provided that $Z\rightarrow \infty$):
\begin{flalign}\label{eq:MDD-THz:TDDEE}
&C_i = \underset{\tau_{\text{CC}},\tau_{\text{CA}},\tau_{\text{AD}}}{\text{maximize}} \ \min\left\{\frac{\tau_{\text{CC}}}{\tau_{\text{CC}}+\tau_{\text{CA}}+2\tau_{\text{GP}}}C_i^{\text{CC}},\frac{\tau_{\text{CA}}}{\tau_{\text{CC}}+\tau_{\text{CA}}+2\tau_{\text{GP}}}C_i^{\text{CA}}, \frac{\tau_{\text{AD}}}{\tau_{\text{CC}}+\tau_{\text{CA}}+2\tau_{\text{GP}}}C_i^{\text{AD}}\right\} \nonumber \\
& \text{s.t.}  \ \tau_{\text{CC}}+\tau_{\text{CA}}+\tau_{\text{AD}}+ \tau_{\text{GP}}\leq 1 \ \text{and} \ \tau_{\text{AD}}\leq \tau_{\text{CC}}+\tau_{\text{CA}}+2\tau_{\text{GP}}. 
\end{flalign}

Set i = i+1\;
}
\KwOut{$C^\ast = \underset{i}{\text{maximize}} \left\{C_i\right\}$}
\end{algorithm}
\vspace{-0.5cm}

\subsection{Frame Structure for TDD-Enabled Fronthaul System}
As for comparison, we also present the TDD frame structure for our proposed fronthaul system, which is shown in Fig. \ref{figure-MDDTHz-TDDF}. Specifically, the TDD transmission frame is divided into $Z$ subframes. Each subframe accounts for $T_{\text{sub}}$ duration including two fronthauls, DL transmission and guard period. Different from MDD-enabled fronthaul, since both two fronthauls in TDD happen over all the subcarriers of the same THz band, they need to be accurately separated for avoiding collision and implemented in a sequential way. Moreover, as DL signals are transmitted over sub-6 GHz band, as shown in Fig. \ref{figure-MDDTHz-TDDF}, the DL transmission of the previous subframe can coexist with the fronthaul transmissions of the current subframe. However, in order to avoid the transmission delay caused by the previous subframe, we assume $\tau_{\text{AD}}\leq \tau_{\text{CC}}+\tau_{\text{CA}}+2\tau_{\text{GP}}$ such that the previous DL transmission can be finished before the current DL transmission starts. The overall end-to-end rate optimization of the TDD-enabled fronthaul system is summarized in Algorithm \ref{MDDTHz:al4}. Compared with MDD, although TDD is not able to balance $C^{\text{CC}}$ and $C^{\text{CA}}$ through adjusting the sizes of $\mathcal{M}_{\text{CC}}$ and $\mathcal{M}_{\text{CA}}$, it can allocate the time durations for two fronthaul and access links to achieve the optimal end-to-end rate.  

\begin{table}
\vspace{-0.5cm}
\caption{Simulation parameters}
\vspace{-0.5cm}
\tiny
\centering
\begin{tabular}{|l|l|}
\hline
Default parameters & Value  \\ \hline
CPU's and AP's Power budget ($P_{\text{CPU}}, P_{\text{AP}}$) & $(35,45)$ dBm \\ \hline
Transmit and Receive antenna gain ($G_t,G_r$) & (20,20) dBi \\ \hline
Absorption coefficient ($k_{\text{abs}}$) & $0.0033 \ \text{m}^{-1}$ \\ \hline
Fresnel reflection coefficient & 0.15  \\ \hline
Roughness factor & $0.088\times 10^{-3}$ \\ \hline
Number of delay taps ($T_{\text{delay}}$) & 6 \\ \hline 
Number of NLoS path ($N_{\text{ray}}$) & 3  \\ \hline
Cyclic prefix length ($N_{\text{CP}}$) & 16  \\ \hline
Sampling rate ($T_{\text{s}}$) & $7.8\times 10^{-12}$ s \\ \hline
NLoS path delay ($\delta$) & $\delta \sim \mathcal{U}[0,N_{\text{CP}} T_{\text{s}}]$ \\ \hline
Azimuth AoD ($\phi$) & $\phi \sim \mathcal{U}(-\pi,\pi)$ \\ \hline
Elevation AoD ($\theta$) & $\theta \sim \mathcal{U}(-\frac{\pi}{2},\frac{\pi}{2})$ \\ \hline
\end{tabular}
\label{Table:MDDTHz:para}\vspace{-0.6cm}
\end{table}

\vspace{-0.3cm}
\section{Simulation Results}\label{sec:MDDTHz:sim}

\vspace{-0.5cm}
\subsection{Parameters and Setup}

In the simulation, we employ the proposed two-tier THz fronthaul for MDD-CF systems, as shown in Fig. \ref{figure-MDDTHz-Archi}, where an indoor industrial workshop with a square area of $(100 \ \text{m} \times 100 \ \text{m})$ and a roof height of $10$ m is considered. The CPU equipped with a $(8\times 8)$ UPA antenna array is installed on the roof at the central point. The $Q=32$ APs equipped with a $(4\times 4)$ UPA antenna array at transmitter and single antenna at receiver are installed on the wall or pillars, and follow uniform distribution within the $xy$-plane, while their height are uniformly distributed in $[4,6]$ m. The $U=8$ single-antenna devices are uniformly distributed on the ground with a fixed height of $1$ m. $U_{\text{max}}$ and $C_{\text{max}}$ are assumed to be 2 and 4, respectively.

For THz fronthaul channels, we assume that the central carrier frequency is $200$ GHz and the allocated bandwidth for fronthaul is $B^{\text{FH}}=1$ GHz including $\left|\mathcal{M}_{\text{FH}}\right|=32$ subcarriers. Then, each subcarrier's bandwidth can be computed as $b^{\text{FH}}=\frac{B^{\text{FH}}}{\left|\mathcal{M}_{\text{FH}}\right|}$. The pulse shaping function used in (\ref{eq:MDD-THz:THzChannelt}) is given by \cite[Eq. (61)]{alkhateeb2016frequency},
where the roll-off factor is set to 1. For sub-6 GHz access channels, the central carrier frequency is assumed to be $5$ GHz with $B^{\text{AD}}=100$ MHz bandwidth. The large-scale fading coefficient $\left\{\beta_{l,q,u}\right\}, \forall l \in \mathcal{L}, q \in \mathcal{C}_l, u \in \mathcal{U}_l$ is given by
$\beta_{l,q,u}[\text{dB}] = -30.5-36.7\log_{10}(d)+\sigma_{\text{sh}}z$ \cite{demir2021foundations},
where $d$ denotes the distance between the $q$-th AP and the $u$-th device, the shadowing is characterized by $\sigma_{\text{sh}}z$ with a standard deviation of $\sigma_{\text{sh}}=4$ dB and $z\sim \mathcal{N}(0,1)$. The noise power is determined based on the bandwidth of fronthaul and access channels following $-174$ dBm/Hz, and hence $\sigma^2_l=\sigma^2_q \neq \sigma^2_u$. We model the SI power as $\sigma^2_{\text{SI}}=\Delta \sigma^2_l$, where the default value of $\Delta$ is -10 dB. Unless otherwise specified, the other parameters are listed in Table I.

\vspace{-0.3cm}
\subsection{Performance Comparison among DC, SC and IDSC}
In this subsection, We will investigate the performance comparison among three AP clustering methods. In Fig. \ref{figure-MDDTHz-MDDDCO}, one random network realization along with the results of AP clustering is presented, where the total numbers of APs and MSs are $Q=32$ and $U=8$, respectively. To be more specific, Fig. \ref{figure-MDDTHz-MDDDCO}(a) shows the 3D distributions of CPU, APs and devices, where the AP grouping and CAP selection are implemented by the DC method. In Fig. \ref{figure-MDDTHz-MDDDCO}(b) and Fig. \ref{figure-MDDTHz-MDDDCO}(c), for the ease of illustration, only the $xy$-plane is plotted without the distribution of devices. As we can observe from Fig. \ref{figure-MDDTHz-MDDDCO}(a) that, with the aid of the DC method, APs can be divided into any number of clusters without the space limitation\footnote{Here, eight AP clusters are used for demonstration. In fact, the DC method can dynamically adjust the AP clusters and finally find the optimal number of AP cluster.}, while SC and IDSC carry out AP clustering relying on sub-areas division. Furthermore, as shown in Fig. \ref{figure-MDDTHz-MDDDCO}(b), since CAPs are fixed at the central points of each sub-area, the AP clustering can not be dynamically changed. As for IDSC, although the placement of CAPs are flexible, it can hardly enable fully dynamic AP clustering like DC. The reason behind is that as the number of sub-areas significantly increases, most of clusters may not have any APs, and hence fail to serve devices. In this case, for the practical implementation, five cases of sub-area division during SC and IDSC are considered in the following simulations, i.e., $L=\left\{4, 8, 12, 16,20 \right\}$.

\begin{figure}[]\vspace{-0.5cm}
\centering
\subfigure[DC]{
\includegraphics[width=0.35\linewidth]{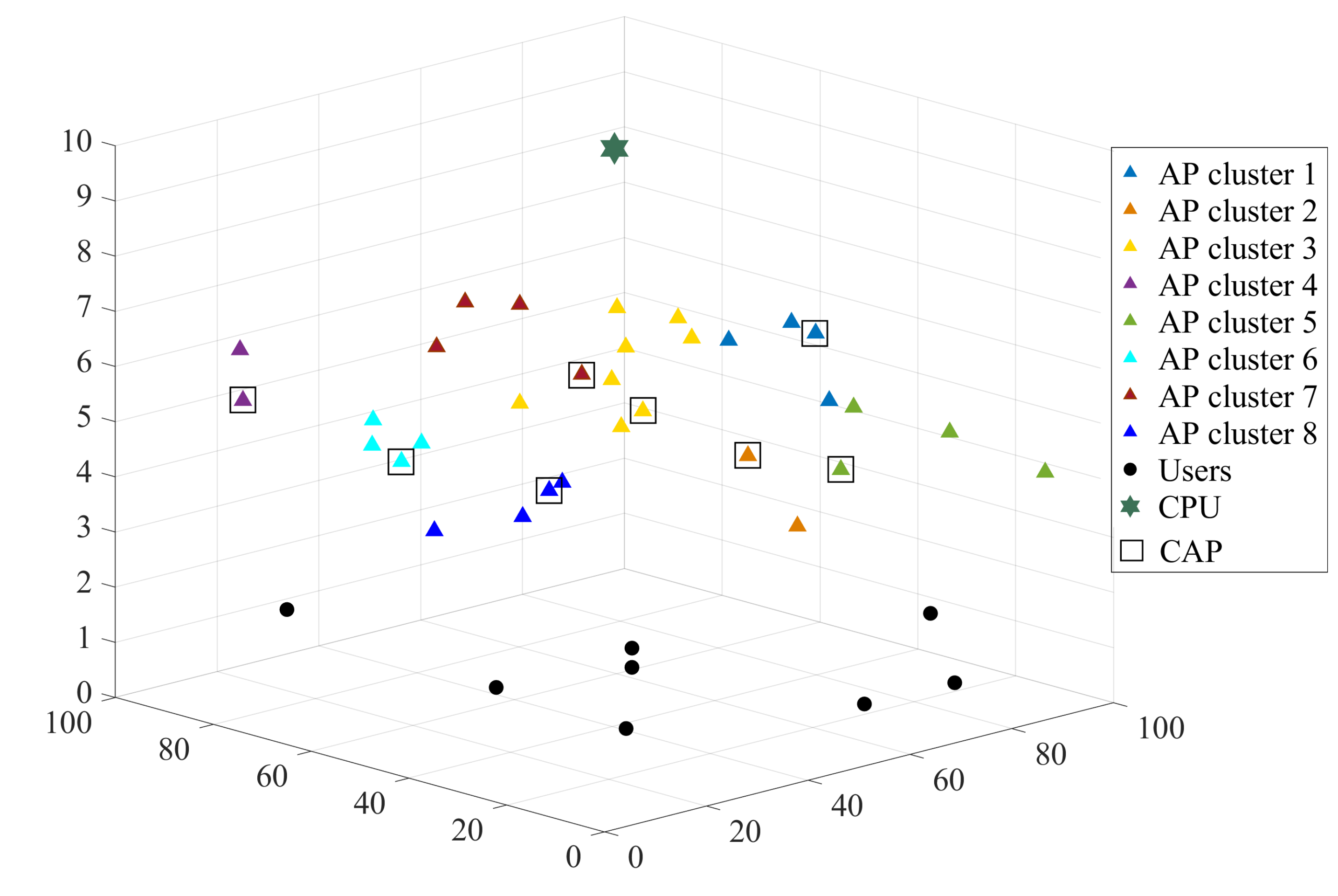}
}
\quad
\subfigure[SC]{
\includegraphics[width=0.25\linewidth]{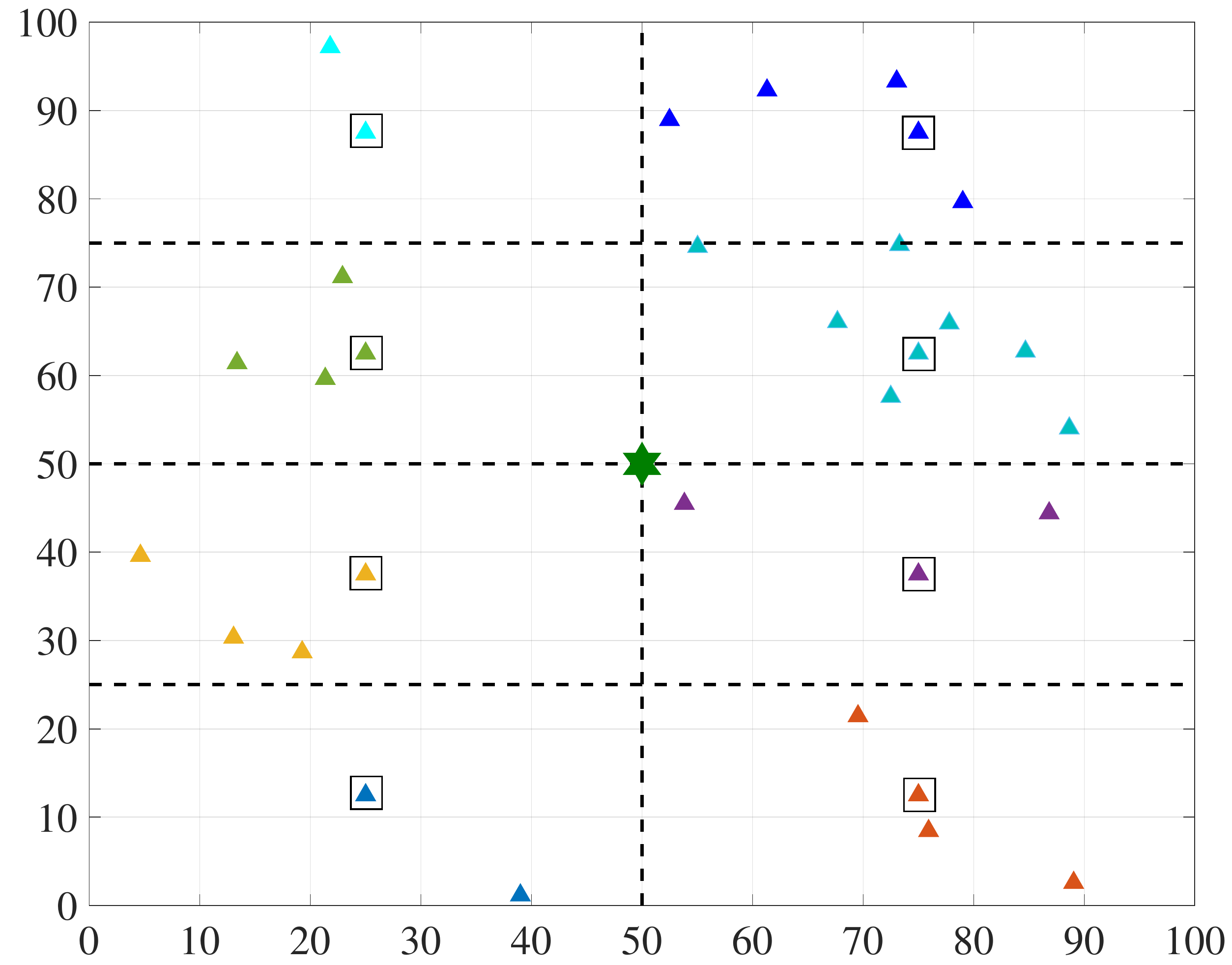}
}
\quad
\subfigure[IDSC]{
\includegraphics[width=0.25\linewidth]{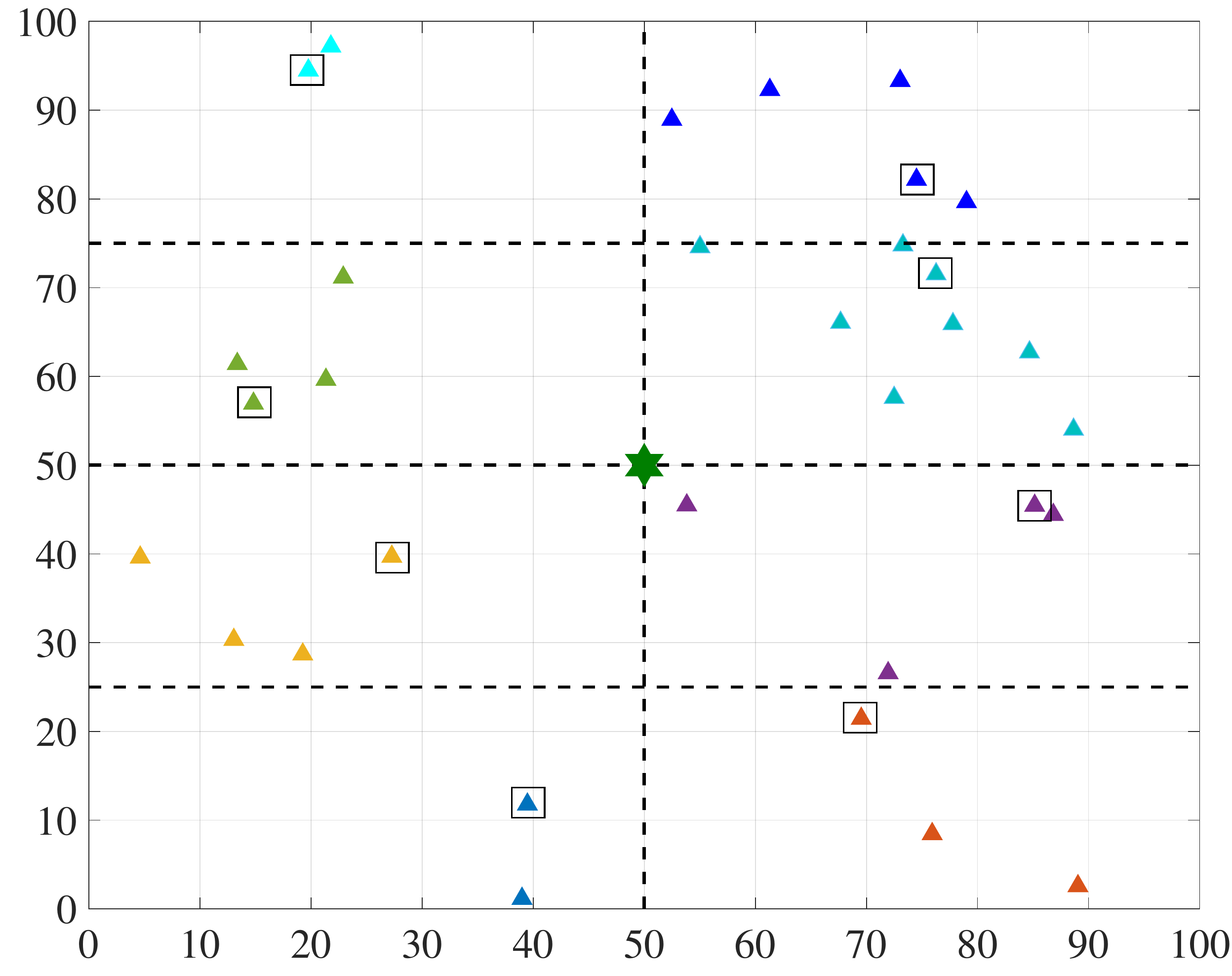}
}
\vspace{-0.5cm}
\caption{The random network realization along with three methods of AP clustering, where $Q=32$, $U=8$ and $L=8$.}
\label{figure-MDDTHz-MDDDCO}\vspace{-0.6cm}
\end{figure}

The performance comparison among DC, SC and IDSC in terms of achievable rates per-device is presented in Fig. \ref{figure-MDDTHz-MDDDSC}, where SC-$l$ denote the SC method is implemented with $l$ AP clusters. It can be observed from Fig. \ref{figure-MDDTHz-MDDDSC}(a) that, when $Q=32$ and $U=8$, DC achieves the similar 90\% likely rates (when the CDF curve is 0.1) with SC-4 and outperforms other methods in terms of median rates (when the CDF curve is 0.5). Especially, DC is more likely to attain the higher rates, with at least 0.2 Gbps rate improvement over other methods. The rationale behind is that with the aid of the {\em K-medoids} method, APs in close proximity can be classified into one cluster, and hence leading to high beamforming gain. However, the situation becomes different for densely distributions APs and devices. e.g., $Q=48$ and $U=16$. As shown in Fig. \ref{figure-MDDTHz-MDDDSC}(b), although DC can still outperform other methods in terms of 90\% likely rates, its 30\% likely rates lags behind SC-16. This is because, with the increasing number of APs, the CAP selection in the DC method may cause the high inter-cluster interference (ICI), when APs belonging to different clusters and staying in close proximity to each other are selected as CAPs\footnote{In the DC method, the CAP selection only considers the LoS complex gains of CPU-to-CAP and CAP-to-AP links, and the potential ICI is neglected. In fact, if ICI is also considered, the weights among two gains and ICI during the CAP selection is of importance, and could be efficiently assigned by deep reinforcement learning method, which is left for our future research.}. Since DC and IDSC use the same method of CAP selection, we can leverage Fig. \ref{figure-MDDTHz-MDDDCO}(c) to further illustrate. As shown in Fig. \ref{figure-MDDTHz-MDDDCO}(c), after the CAP selection, the CAPs of upper-right two clusters situate close to each other, leading to the severe ICI problem. In contrast, CAPs in SC method are always placed at the central points of each sub-area, which can mitigate the ICI to some extent. In addition, since IDSC is not only subject to the rigid AP clustering method, but also has the same ICI problem with DC, it achieves a poor performance in both cases, as observed from Fig. \ref{figure-MDDTHz-MDDDSC}(a) and Fig. \ref{figure-MDDTHz-MDDDSC}(b). It can be concluded that, although SC method with a certain pattern of sub-area division is more likely to attain the higher per-device rates, it fails to be adaptive to the varying sizes of networks. Consequently, the DC method can be deemed as a more promising approach for our proposed two-tier wireless fronthaul systems. 

\begin{figure}[]\vspace{-0.5cm}
\centering
\subfigure[]{
\includegraphics[width=0.25\linewidth]{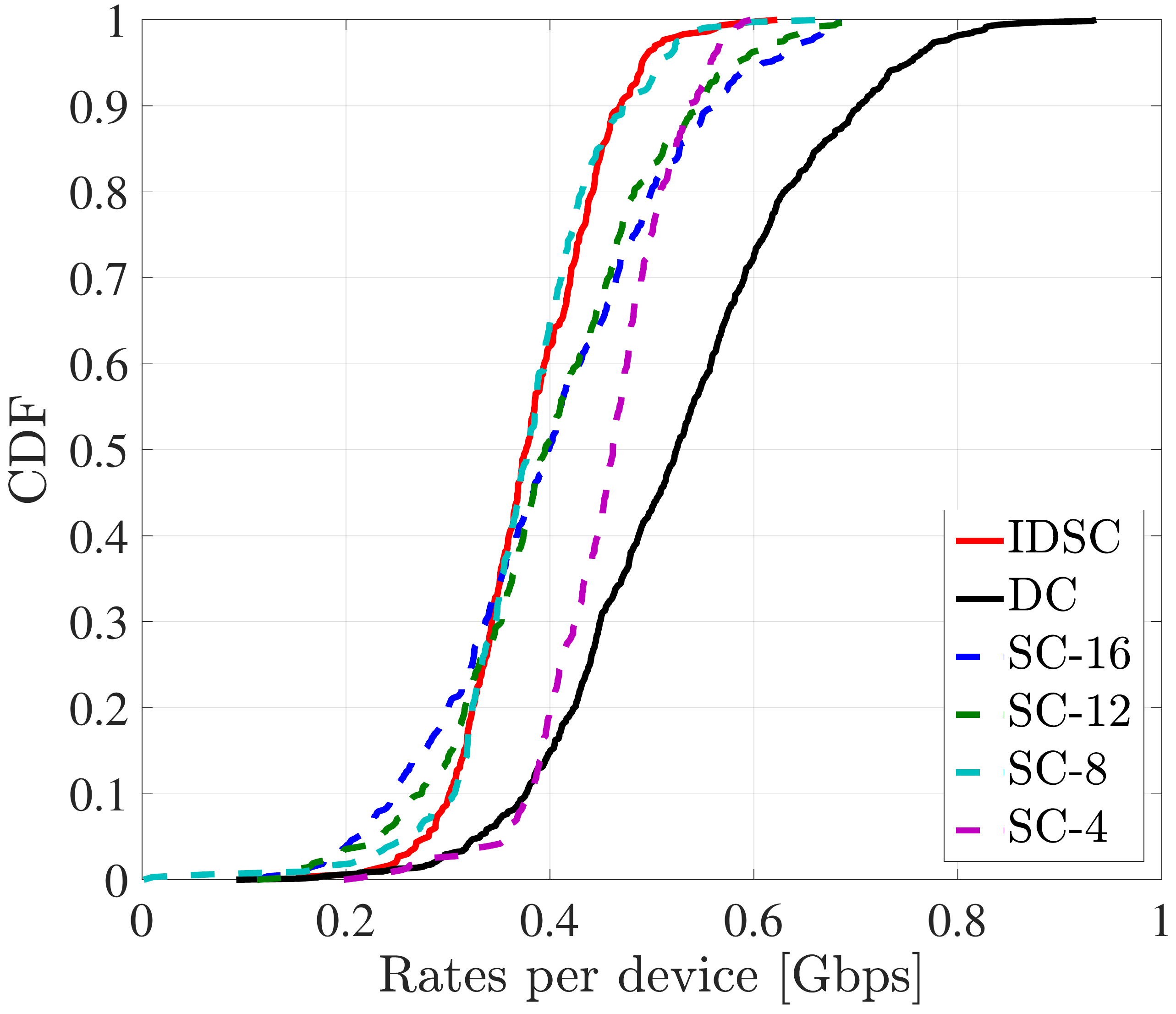}
}
\quad
\subfigure[]{
\includegraphics[width=0.255\linewidth]{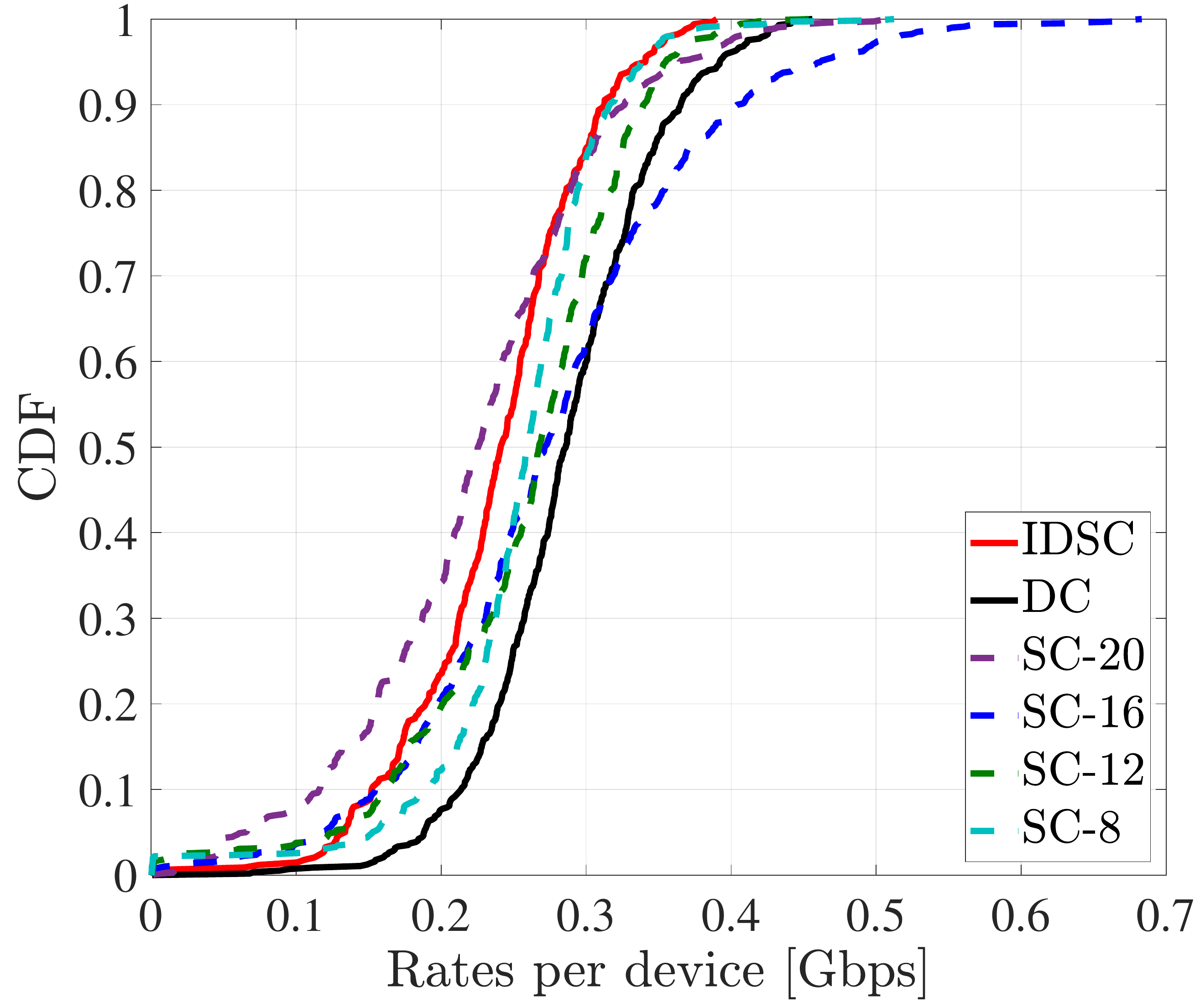}
}
\vspace{-0.5cm}
\caption{Performance comparison between DC, SC and IDSC in terms of achievable rates per-device. (a): $Q=32$ and $U=8$; (b): $Q=48$, $U=16$, and note that as $U_{\text{max}}=2$, the minimum $L$ is set to $8$ for SC method so as to guarantee that each device can be served with at least one AP cluster. Moreover, for fair comparison, SC-20 is added for case (b).}
\label{figure-MDDTHz-MDDDSC}\vspace{-0.8cm}
\end{figure}

Fig. \ref{figure-MDDTHz-MDDDAP}(a)-(b) shows the end-to-end and fronthaul/access rates per-device achieved by the DC method under different numbers of APs and THz bandwidth. To be more specific, in Fig. \ref{figure-MDDTHz-MDDDAP}(a), two-tier fronthaul with 1 GHz bandwidth achieve much lower rates than the AP-to-Device access link. In this case, the end-to-end rates mainly depends on the fronthaul rates. Accordingly, as observed from Fig. \ref{figure-MDDTHz-MDDDAP}(b), the increasing APs may impose the burden on the fronthaul links, and hence lead to the decreased end-to-end rates. In order to lift the end-to-end rate, we increase the THz bandwidth to 4 GHz, and it can be seen from Fig. \ref{figure-MDDTHz-MDDDAP}(a) that, equipped with the sufficient bandwidth, the fronthaul rates are larger than the access rates, and the end-to-end rates of the system can be further enhanced by increasing the number of APs, owing to the extra beamforming gains, as shown in Fig. \ref{figure-MDDTHz-MDDDAP}(b).

Fig. \ref{figure-MDDTHz-MDDDAP}(c)-(d) shows the end-to-end and fronthaul/access rates per-device achieved by the DC method under different numbers of devices and THz bandwidth. Similar to Fig. \ref{figure-MDDTHz-MDDDAP}(a)-(b), the system employed with larger THz bandwidth is able to achieve better performance in terms of fronthaul/access and end-to-end rates. However, a different observation is that with the increasing number of devices, both the fronthaul and access rates decrease, since the available resource such as subcarriers and power for each device are reduced. In particular, when $U$ is increased to 20, it can be observed from Fig. \ref{figure-MDDTHz-MDDDAP}(c) that although the larger THz bandwidth can significantly lift the performance of fronthaul link, the access link is highly affected by the dense distribution of devices due to the severe interference, and hence the improvement of overall end-to-end rates by allocating more THz bandwidth is limited when the number of devices is very large, as shown in Fig. \ref{figure-MDDTHz-MDDDAP}(d). In this regard, we can increase the number of APs to provide more beamforming gains, or apply NOMA to mitigate the inference caused by the neighboring AP clusters, which will be left for our future research.   

\begin{figure}[]\vspace{-0.3cm}
\centering
\subfigure[Fronthaul and access rates, where $U=8$]{
\includegraphics[width=0.2\linewidth]{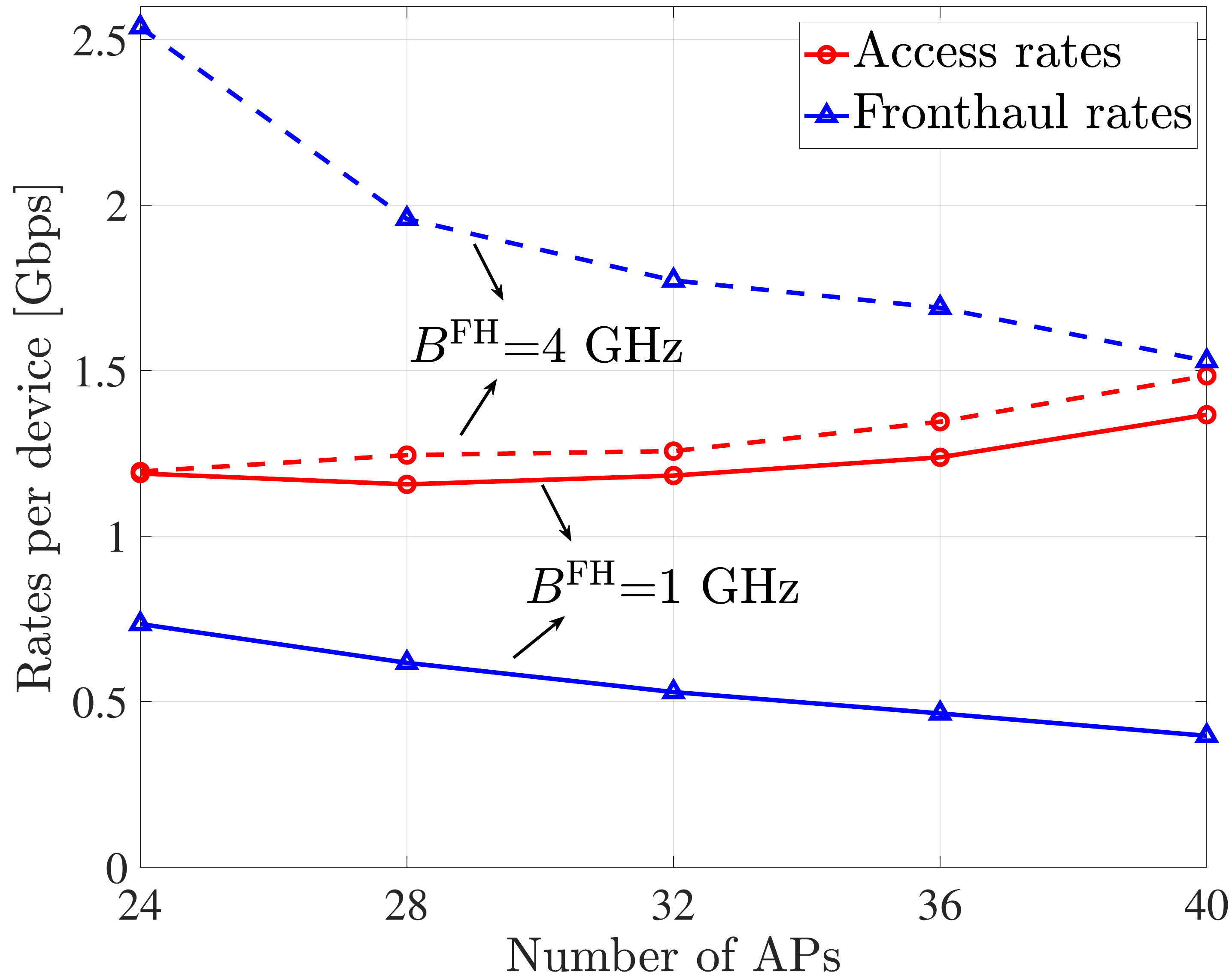}
}
\quad
\subfigure[End-to-end rates, where $U=8$]{
\includegraphics[width=0.2\linewidth]{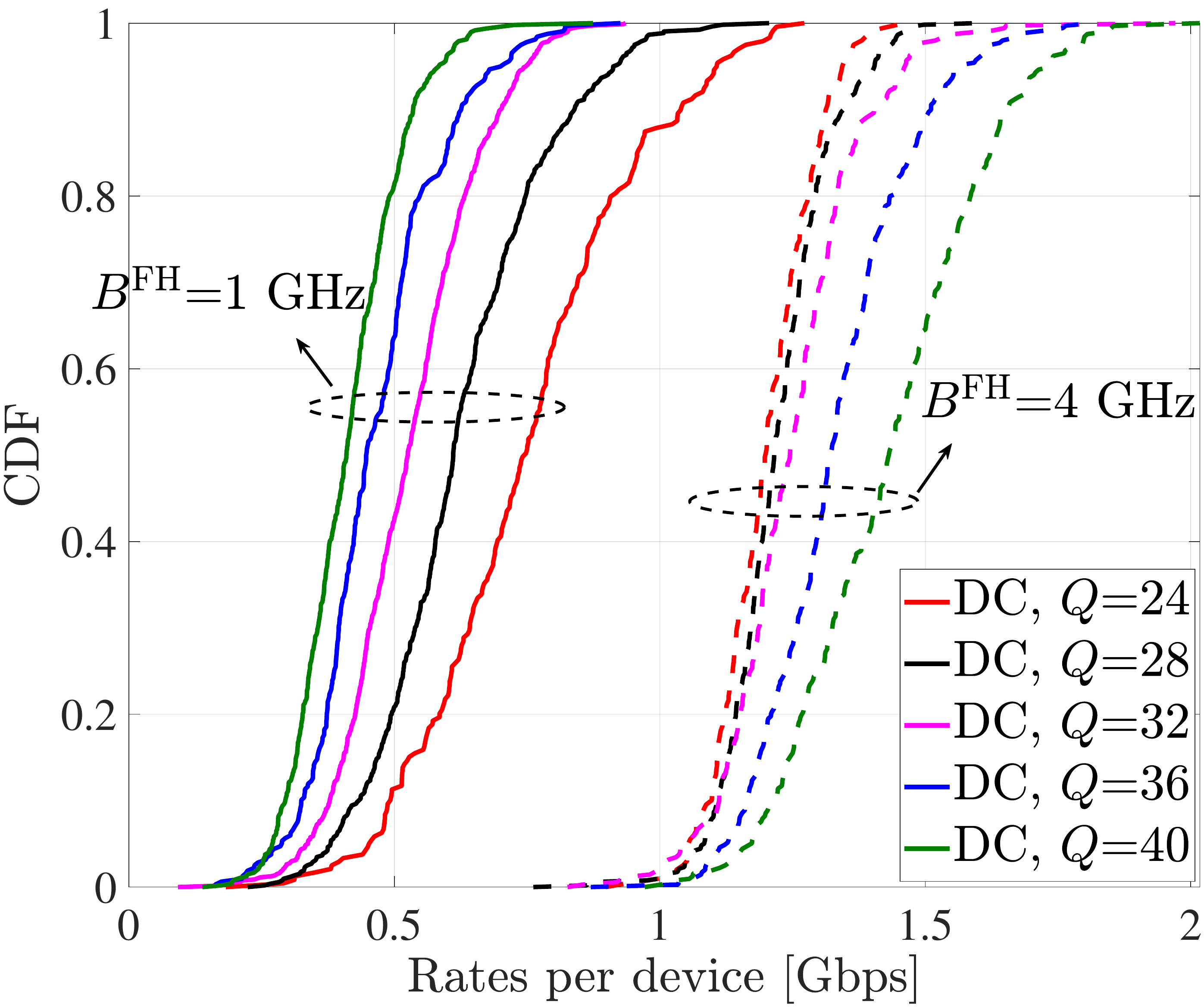}
}
\quad
\subfigure[Fronthaul and access rates, where $Q=32$]{
\includegraphics[width=0.2\linewidth]{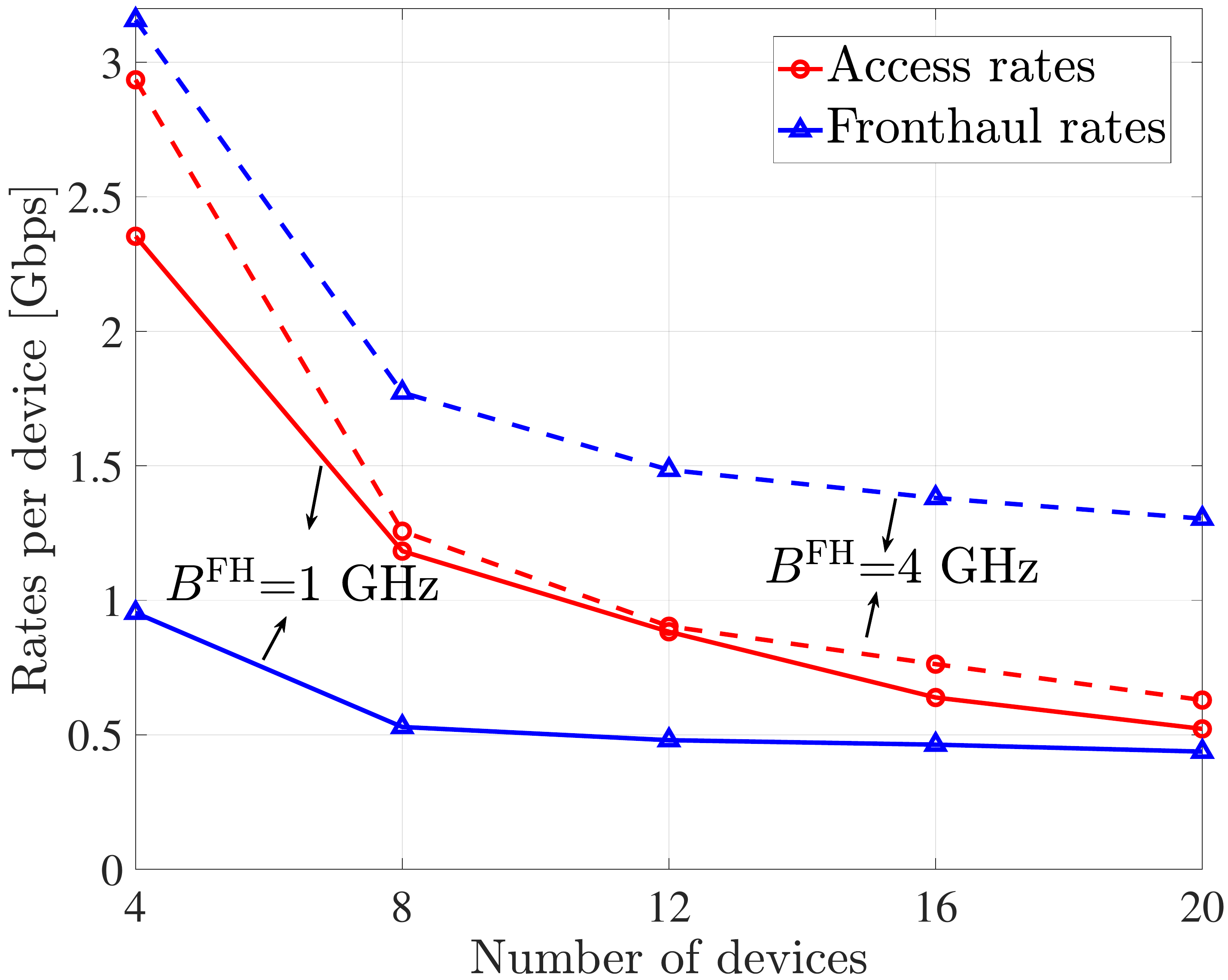}
}
\quad
\subfigure[End-to-end rates, where $Q=32$]{
\includegraphics[width=0.2\linewidth]{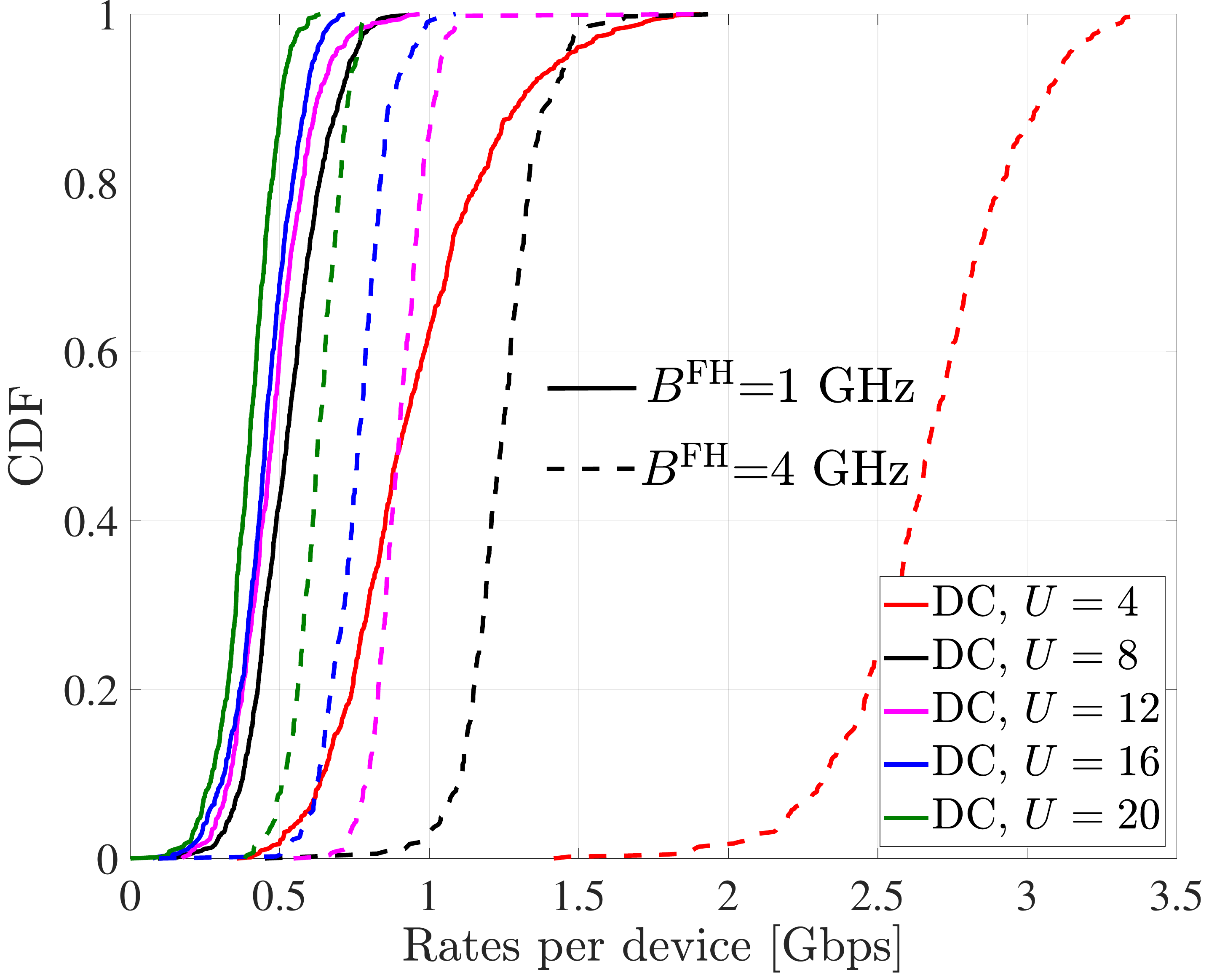}
}
\vspace{-0.3cm}
\caption{The fronthaul/access and overall end-to-end rates achieved by the proposed DC method with different numbers of APs/devices and THz bandwidth.}
\label{figure-MDDTHz-MDDDAP}\vspace{-0.8cm}
\end{figure}

\vspace{-0.3cm}
\subsection{Performance Comparison among Different Fronthaul Schemes}
Firstly, we compare MDD and TDD modes in our proposed two-tier wireless fronthaul system. The default value of $\tau_{\text{GP}}$ in the TDD frame structure is assumed to be 0.05 according to \cite{3gpp2017nr,etsi2013136}\footnote{The proportion of guard period in the TDD frame structure depends on the subframe configuration \cite{etsi2013136}. Here, we set $\tau_{\text{GP}}=0.05$ to minimize the influence of the guard period on the evaluation of TDD system's performance.}. From Fig. \ref{figure-MDDTHz-MDDDTDD}(a), when $\Delta\leq 20$ dB meaning that the power of the residual SI is 100 times less than that of noise, MDD is capable of outperforming TDD, owing to the well-designed parallel frame pattern and the non-GP configuration. Moreover, if $\Delta$ is increased to 30 dB, the large power of residual SI makes MDD lose its advantages over TDD. It is worth mentioning that compared with the SI cancellation in the conventional IBFD-based sub-6G systems, the mitigation of SI could be much easier in MDD-based THz systems. The reason is that the path loss of LoS in the THz band is extremely large, not to mention that of NLoS. Hence, the most of SI over THz band can be significantly suppressed in the propagation domain by using the classic methods such as adding absorber or blockages between transmitter and receiver \cite{kolodziej2019band}. Then, as desired signal and SI are over orthogonal subcarriers, MDD is able to further cancel the residual SI in the digital domain. In this case, SI is expected to be suppressed under the noise floor (i.e., $\Delta\leq 0$ dB), which makes MDD a promising solution for two-tier wireless fronthaul systems. We further study the impact of $\tau_{\text{GP}}$ in TDD-based systems in Fig. \ref{figure-MDDTHz-MDDDTDD}(b), as expected, TDD is heavily subject to the proportion of guard period. In particular, when $\tau_{\text{GP}}=0.1$, the 90\% likely rates per device of TDD is 0.1 Gbps less than that of MDD. Moreover, it can be also observed that if TDD removes guard period, it may achieve a comparable performance to MDD.

\begin{figure}[]\vspace{-0.5cm}
\centering
\subfigure[]{
\includegraphics[width=0.25\linewidth]{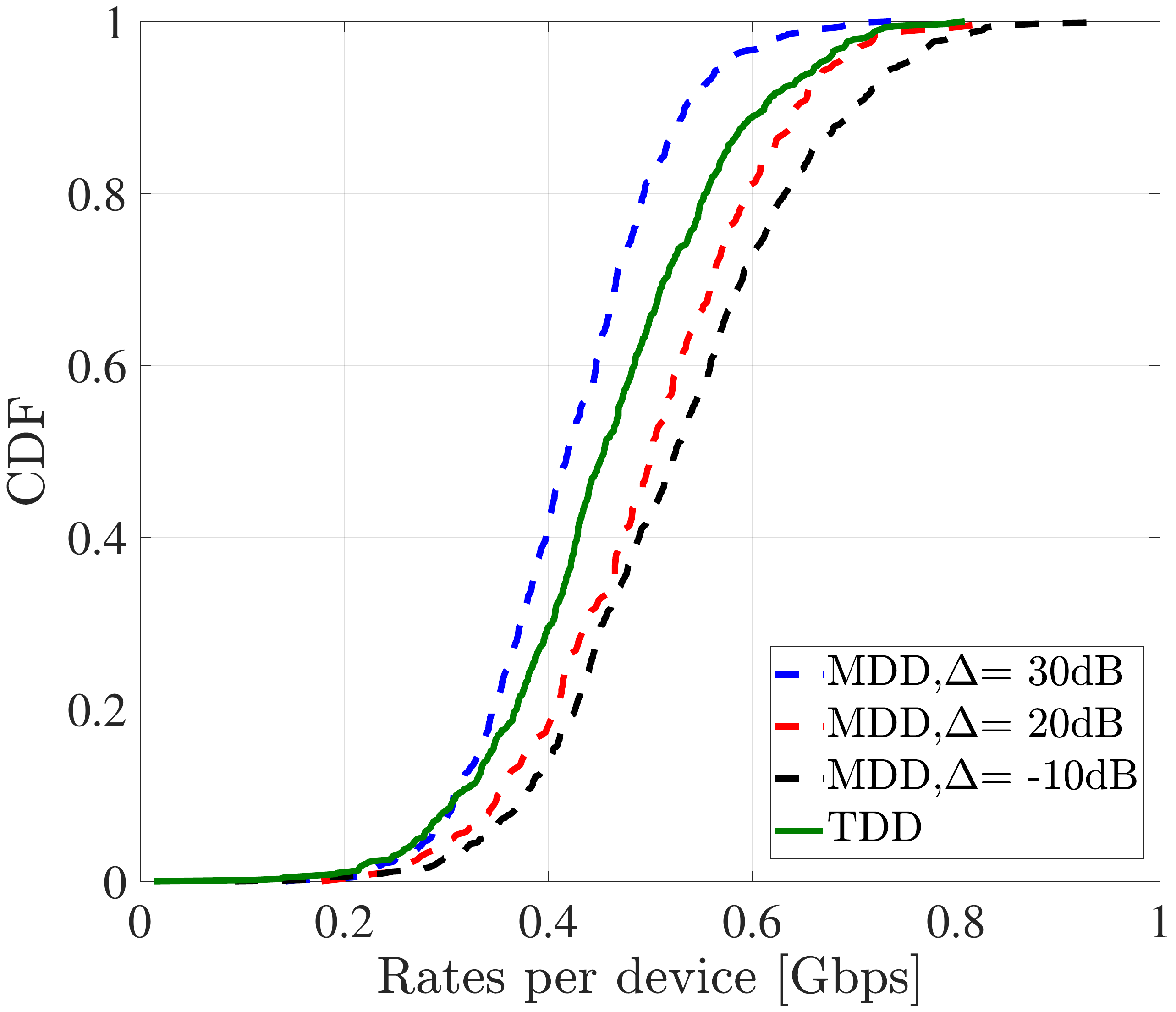}
}
\quad
\subfigure[]{
\includegraphics[width=0.25\linewidth]{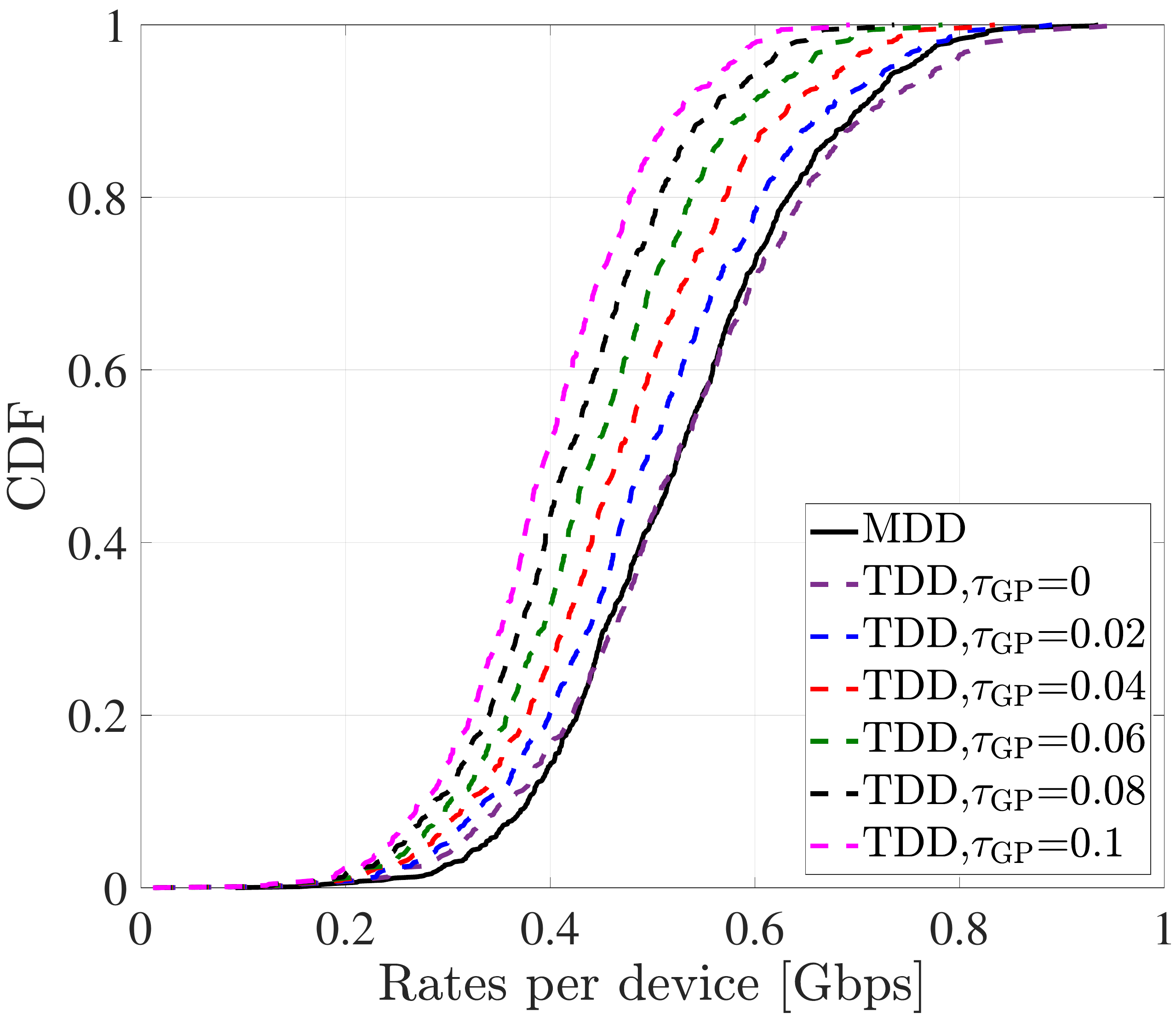}
}
\vspace{-0.5cm}
\caption{Performance comparison between MDD and TDD in our proposed two-tier wireless fronthaul architecture.}
\label{figure-MDDTHz-MDDDTDD}\vspace{-0.8cm}
\end{figure}

Obviously, compared with the traditional fiber-optic or radio stripes fronthaul, the proposed THz-enabled two-tier wireless fronthaul is capable of providing the highest flexibility in deploying APs in an indoor industrial scenario. Nevertheless, the large path loss and interference may cause the reduced fronthaul rates. To address this problem, the abundant bandwidth in THz band can be leveraged to close the performance gap with wired fronthaul schemes.
In this regard, we will make a comprehensive performance comparison among different fronthaul schemes, and verify if our proposed scheme can be comparable to the wired fronthaul schemes in terms of achievable rates in cell-free indoor industrial systems. The following benchmarks are considered:
\begin{itemize}
\item TDD-TTWL: This scheme denotes the TDD-based two-tier wireless fronthaul.
\item CC-HY: According to \cite{demirhan2022enabling}, CC-HY denotes a hybrid fronthaul scheme that the CPU-to-CAP links are operated over the THz band, while each AP within the cluster is connected with the CAP via fiber-optic cables.
\item CA-HY: Inspired by \cite{ammar2021downlink}, in the CA-HY scheme, the CAP-to-AP links within each cluster are based on THz fronthaul, and the wired links connect CPU and CAPs.
\item TTW: This scheme denotes the two-tier wired fronthaul proposed in \cite{zhang2022user}, where both CPU-to-CAP and CAP-to-AP links are implemented over fiber-optic cables. 
\item STWL: This scheme denotes the single-tier wireless fronthaul, where CPU directly sends fronthaul signal to every AP over THz band and the AP clustering is no longer needed.
\item STW: Contrary to STWL, STW enables the single-tier wired connections between CPU and all the APs.
\end{itemize}
Note that apart from STWL and STW, all the other schemes implement AP clustering, and the optimal $L$ is determined by following step 2-18 in Algorithm \ref{MDDTHz:al3} without the process of subcarrier assignment (that is to say, the $\mathcal{M}_{\text{CC}}$ and $\mathcal{M}_{\text{CA}}$ assignment is exclusive in the MDD-based schemes).

\begin{figure}
\centering
\includegraphics[width=0.3\linewidth]{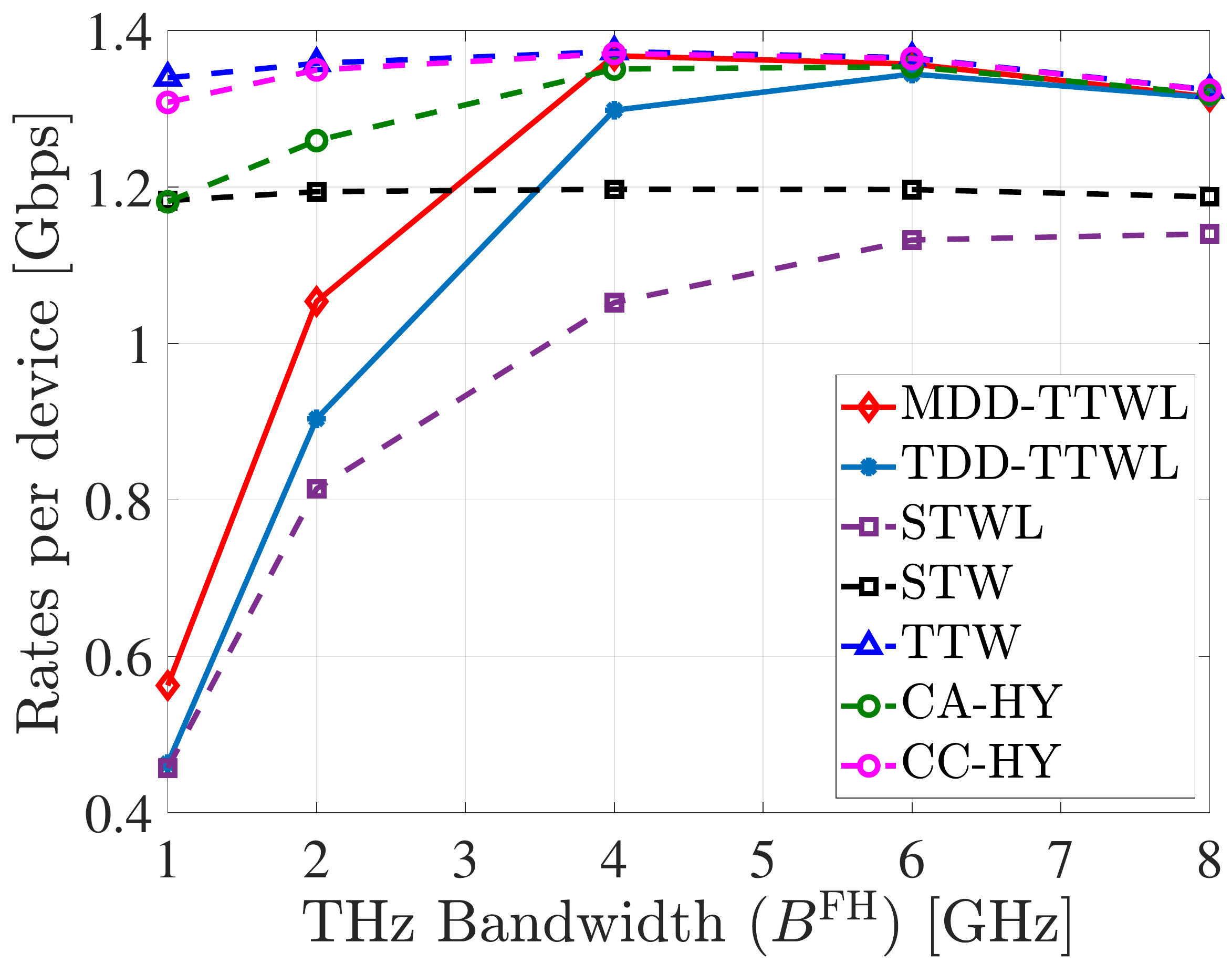}
\vspace{-0.5cm}
\caption{Performance comparison among different fronthaul schemes versus the THz bandwidth.}
\label{figure-MDDTHz-MDDVs}\vspace{-0.8cm}
\end{figure}

From Fig. \ref{figure-MDDTHz-MDDVs}, as expected, since TTW scheme applies fiber-optic cables to both CPU-to-CAP and CAP-to-AP links, it achieves the highest rates per device. As the performance of wireless CAP-to-AP links are affected by not only the inter-AP interference within the cluster, but also the inter-cluster interference, and hence CA-HY lags behind CC-HY. Furthermore, MDD-TTWL outperforms TDD-TTWL owing to the parallel frame structure and free of the guard period, as we analyzed in Fig. \ref{figure-MDDTHz-MDDDTDD}. In particular, MDD-TTWL is able to attain the similar rates as TTW when $B^{\text{FH}}$ is increased to 4 GHz. The reason is that as the THz bandwidth increases, the overall system performance mainly depends on the AP-to-Device access links over sub-6 GHz band. In this case, it was shown that the sufficient THz bandwidth makes the proposed two-tier wireless fronthaul scheme feasible in indoor industrial CF systems. On the other hand, without the optimization of AP clustering, STW and STWL cannot achieve a comparable performance to other fronthaul schemes. More specifically, in the STWL scheme, since the CPU communicates with all the APs over the THz band, the AP far away from the CPU experiences the undesirable fronthaul connection leading to the low overall end-to-end rates.

\vspace{-0.3cm} 
\subsection{Ray-Tracing Based Simulations in Practical Scenario}

We further evaluate our proposed MDD-enabled two-tier THz fronthaul in a practical indoor industrial scenario, as shown in Fig. \ref{figure-MDDTHz-RTscen}(a)-(b). This scenario corresponds to the digital twin obtained from LiDAR scans in a BOSCH factory in Blaichach, Germany. Electromagnetic properties have been assigned to the different objects in the scenario to obtain realistic frequency dependent simulations. The ray-tracing simulations were performed with the software WinProp. The height of CPU is 6.9 m, while the heights of all the APs and devices are 3 m and 1.5 m, respectively. A rich multipath propagation environment can be observed in the isotropic power delay profiles (PDP) depicted in Fig. \ref{figure-MDDTHz-RTscen}(c) (for sake of simplicity only the results observed at the AP 8 receiver are displayed). The difference of bandwidth between access and THz channels is observed in the resolution in the delay domain. From the figure, it can be observed that subject to the feature of high frequency, the received power over the THz channels attenuate heavily, especially when two nodes stay away from each other and there is no LoS path between them.  

\begin{figure}[]\vspace{-0.5cm}
\centering
\subfigure[]{
\includegraphics[width=0.65\linewidth]{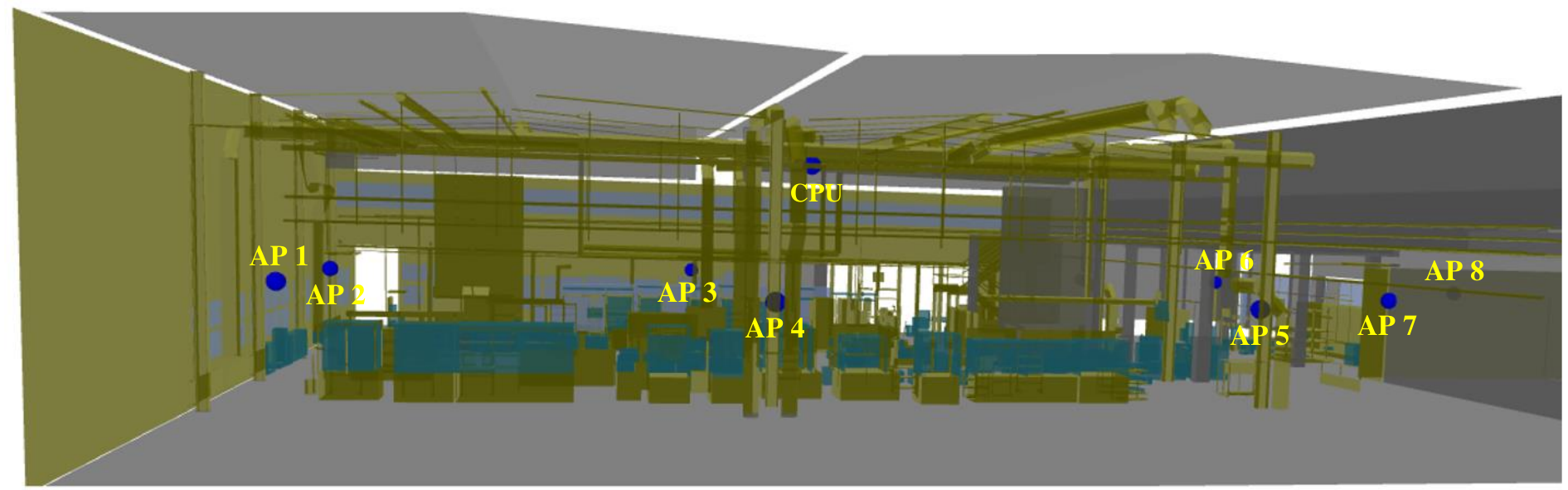}
}
\quad
\subfigure[]{
\includegraphics[width=0.32\linewidth]{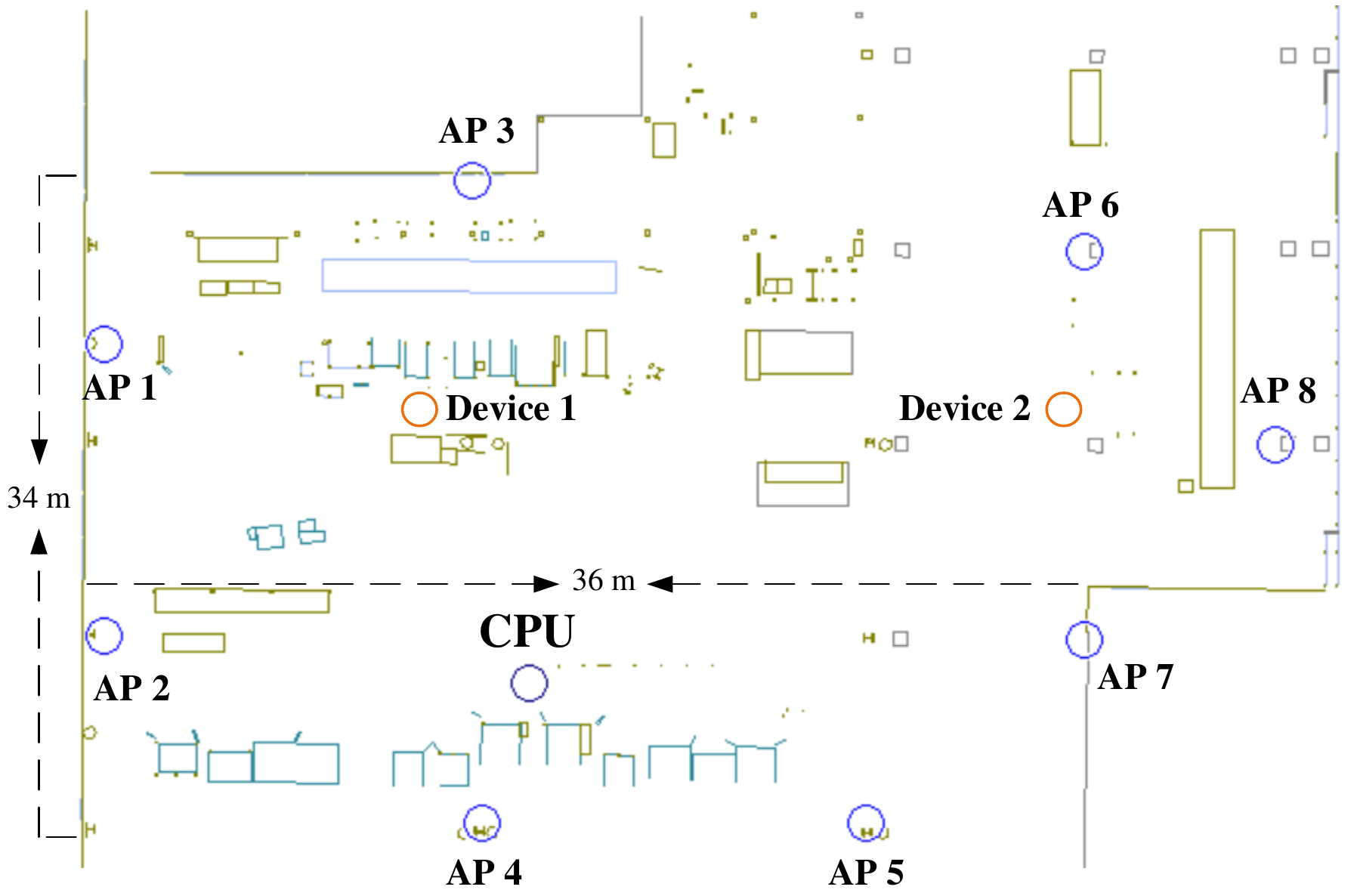}
}
\quad
\subfigure[]{
\includegraphics[width=0.25\linewidth]{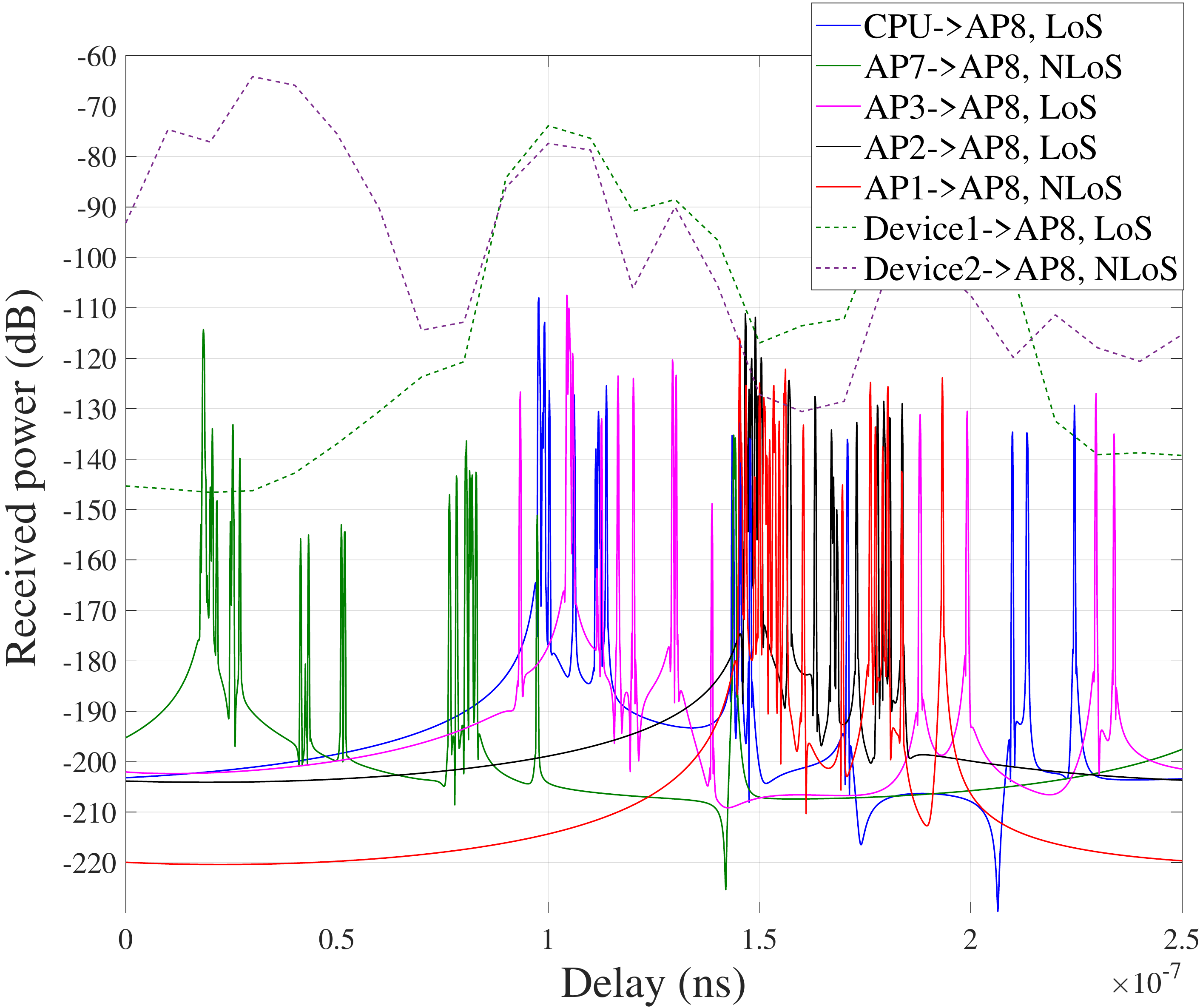}
}
\quad
\subfigure[]{
\includegraphics[width=0.29\linewidth]{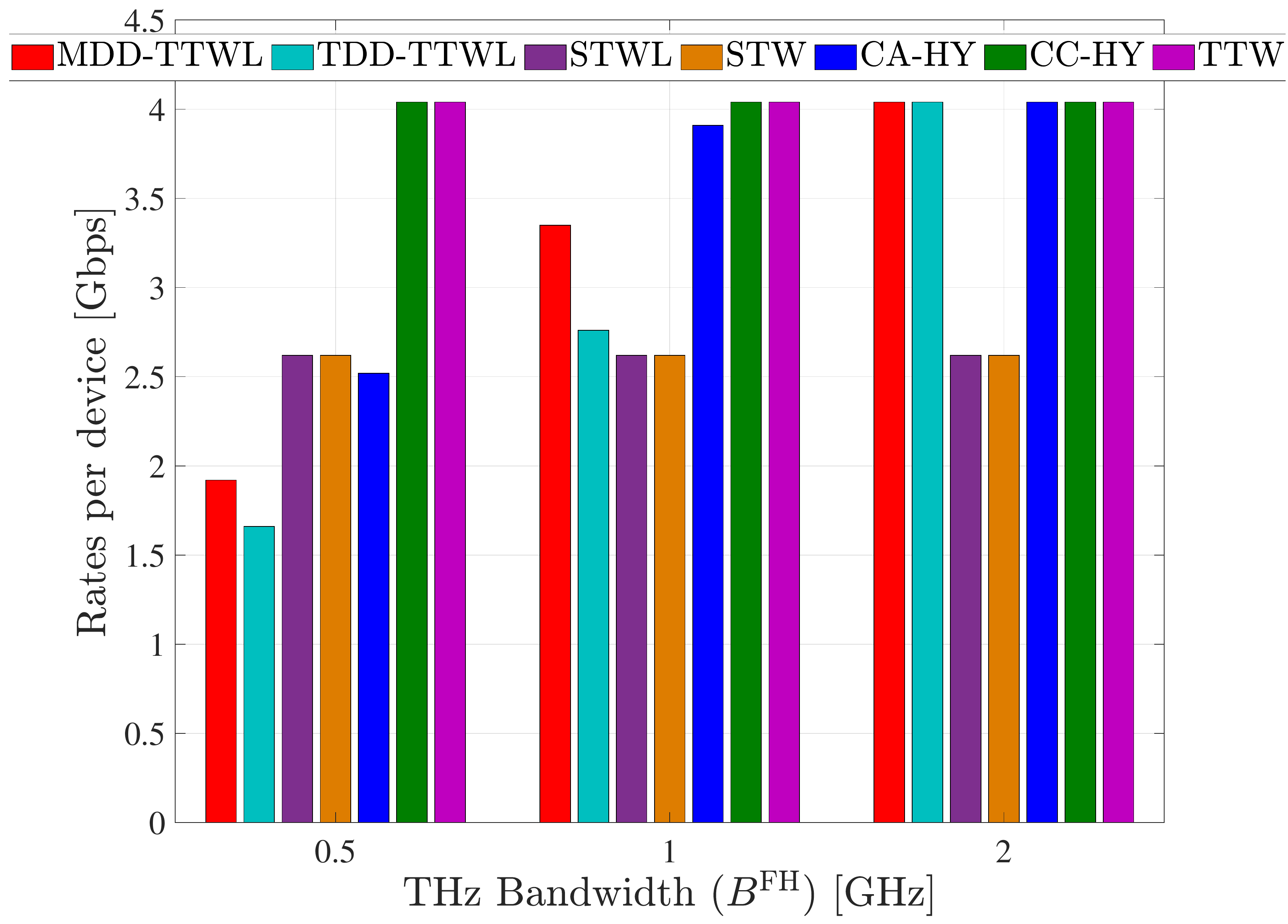}
}
\vspace{-0.5cm}
\caption{(a) A practical indoor industrial scenario, where the CPU is installed on the roof, while all the APs are installed on the wall or pillars. (b) The corresponding ray-tracing environment.
(c) Power delay profile measured at AP 8 receiver. (d) Performance comparison among different fronthaul schemes versus the THz bandwidth in practical indoor industrial scenario.}
\label{figure-MDDTHz-RTscen}\vspace{-0.8cm}
\end{figure}

Furthermore, we obtain the CIR of both access and THz channels based on ray-tracing channel measurements. Except the channel related parameters, all the other parameters (such as power budget, antenna gain and noise/SI power) are kept the same as in Section VI-A, and then a performance comparison among various fronthaul schemes in this practical scenario is presented in Fig. \ref{figure-MDDTHz-RTscen}(d). As expected, the results match with that in Fig. \ref{figure-MDDTHz-MDDVs}. Specifically, the MDD-TTWL outperforms TDD-TTWL when the THz bandwidth is less than 2 GHz, and CC-HY and TTW achieve the upper bound of achievable rates. Moreover, when the available bandwidth of THz-based fronthaul channels increases to 2 GHz, MDD-TTWL can attain the same rates as CC-HY and TTW. This observation implies that for the consideration of ease deployment and robust achievable rates, our proposed MDD-enabled two-tier THz scheme can be deemed as a promising approach to enabling fully-wireless fronthaul in practical indoor industrial scenarios.  

\section{Conclusion}
In this paper, we proposed an MDD-enabled two-tier THz fronthaul scheme, which can make the indoor industrial CF-mMIMO systems free from the costly wired fronthaul links.
 To maximize the achievable rates in such a complicated scheme, we firstly resorted to the low-complexity but efficient heuristic methods to relax binary variables. Then, the iterative optimization of the assignment of subcarrier sets and the number of AP clusters was applied to obtain the optimal overall end-to-end rates. Furthermore, an advanced MDD frame structure was tailored for the proposed scheme. Our studies and numerical results showed that compared with TDD, MDD is a potentially promising solution to enable two-tier THz fronthaul scheme. Moreover, with the aid of the sufficient THz bandwidth, our proposed scheme is capable of not only achieving the same rates with the existing fiber-optic based CF-mMIMO systems, but also imposing less burden on the system overhead. In order to further demonstrate the practicality of the proposed scheme, we applied it to a real indoor industrial scenario relying on ray-tracing based channel measurements, and the simulation results showed its superior over other methods.

\ifCLASSOPTIONcaptionsoff
  \newpage
\fi



%

\bibliographystyle{IEEEtran}
\bibliography{MDD_bib_short}

%








\end{document}